\documentclass[12pt]{article}
\usepackage[margin=1in]{geometry}

\usepackage{amsthm}
\usepackage{amssymb}
\usepackage{amsmath}
\usepackage{graphicx}
\usepackage{wrapfig}
\usepackage{bbm}
\usepackage{float}
\usepackage{tikz}
\usetikzlibrary{arrows,calc,shapes,decorations.pathreplacing,patterns}
\usepackage{hyperref}
\hypersetup{pdftex,colorlinks=true,allcolors=blue}
\usepackage{hypcap}
\usepackage{caption}
\usepackage{pgfplots}
\usepackage{subcaption}
\usepackage{enumitem}

\def\addlegendimage{\csname pgfplots@addlegendimage\endcsname}

\newtheorem{theorem}{Theorem}
\newtheorem{lemma}[theorem]{Lemma}
\newtheorem{proposition}[theorem]{Proposition}
\newtheorem{corollary}[theorem]{Corollary}

\newtheorem{assumption}[theorem]{Assumption}

\newtheorem*{theorem*}{Theorem}
\newtheorem*{remark}{Remark}

\theoremstyle{definition}
\newtheorem{example}{Example}

\def\P{\mathop{\mathrm{P}}}
\def\E{\mathrm{E}}

\newcommand{\floor}[1]{\left\lfloor #1 \right\rfloor}
\newcommand{\ceil}[1]{\left\lceil #1 \right\rceil}

\begin{document}

\title{A strategic model of job arrivals to a single machine with earliness and tardiness penalties\footnote{To appear in The IISE Transactions}}

\author{Amihai Glazer\footnote{Department of Economics, University of California, Irvine, aglazer@uci.edu}, Refael Hassin\footnote{Department of Statistics and Operations Research, Tel Aviv University, hassin@post.tau.ac.il} and Liron Ravner\footnote{Department of Statistics and Operations Research, Tel Aviv University, lravner@post.tau.ac.il}}

\date{\today}
\maketitle

\begin{abstract}

We consider a game of decentralized timing of jobs to a single server (machine) with a penalty for deviation from a due date, and no delay costs. The jobs' sizes are homogeneous and deterministic. Each job belongs to a single decision maker, a customer, who aims to arrive at a time that minimizes his deviation penalty. If multiple customers arrive at the same time then their order of service is determined by a uniform random draw. We show that if the cost function has a weighted absolute deviation form then any Nash equilibrium is pure and symmetric, that is, all customers arrive together. Furthermore, we show that there exist multiple, in fact a continuum, of equilibrium arrival times, and provide necessary and sufficient conditions for the socially optimal arrival time to be an equilibrium. The base model is solved explicitly, but the prevalence of a pure symmetric equilibrium is shown to be robust to several relaxations of the assumptions: restricted server availability, inclusion of small waiting costs, stochastic job sizes, randomly sized population, heterogeneous due dates, and non-linear deviation penalties.
\end{abstract}

\section{Introduction}\label{sec:intro}

This paper complements the classical machine timing/sequencing problem with a common due date and earliness and tardiness penalties. Prior analyses focused on centralized analysis, that is, a single decision maker determining the schedule for all jobs. We extend that work by analyzing a decentralized setting in which each job is a rational agent; the centralized socially optimal schedule is a benchmark for comparison. 

More specifically, we study a single server system with customers sending their jobs to be processed in a single-machine shop according to a first-come first-served (FCFS) order. Jobs that arrive simultaneously are randomly ordered. Customers have desired due dates for processing their jobs; but the actual start time depends on their arrival times relative to the others, as dictated by the FCFS regime. Thus a game results where the players are the customers, the strategies are the arrival times to the system, and each player aims to minimize his expected earliness and tardiness penalties. Our model extends the vast literature on scheduling jobs subject to due dates, by assuming that jobs belong to individual agents who are rational and act strategically to maximize their utilities. 

We assume that the number of customers is finite and each has a single job. When some customers arrive together they are ordered according to a uniform random draw. If all customers have a common due date, or close due dates, then arriving together with other customers is potentially beneficial because in expectation every customer is in the ``middle," which may be closer to the due date than arriving before or after everyone else. If the job sizes of the customers are deterministic and the cost function is a weighted absolute deviation, i.e. some linear penalty for early service and a possibly different linear penalty for late service, we find that simultaneous arrivals (clustering) is the only possible equilibrium. The equilibrium arrival time is not unique, but rather there is a whole interval such that if all customers arrive at any time within that interval then no customer has an incentive to deviate. Moreover, the socially optimal arrival time often lies within this interval and is therefore an equilibrium. Specifically this holds if the earliness and tardiness penalties are not too different. Simple explicit necessary and sufficient conditions are given below.

A general framework for classical machine timing/sequencing problem with a single server and weighted penalties is presented and analysed in \cite{HP1991}. Surveys of the research on these problems can be found in \cite{BS1990} and \cite{HS2005}. Our work also relates to the research on decentralized multi-machine routing games (\cite{BH2007},\cite{FT2012} and \cite{AFJMS2015}), where the decision variable for the individual jobs is the choice of machine and not a timing decision.

Our base model assumes linear penalties, deterministic population and job sizes, no waiting costs, common due dates and unrestricted machine availability. It is thus  tractable and yields explicit results. The equilibrium solution of simultaneous arrivals, however, is valid for a more general decentralized job timing game. Specifically the due dates need not be common, but rather close enough, the cost function needs to be unimodal, there may be small waiting costs and the job sizes and population sizes may be random. We explore all of these possibilities. We further show how the equilibrium outcome is modified if the server is only available for a restricted period of time. We also present examples for which a symmetric pure-strategy equilibrium does not exist. Relaxing the assumptions of the base model is important to show that the results are robust for more realistic machine scheduling settings (heterogeneous due dates or deviation penalties and restricted availability), and for other applications such as queueing (stochastic population and job sizes) and transportation (waiting costs).

The next section reviews relevant literature. Section \ref{sec:model} defines the base model. This is followed by deriving the socially optimal solution in Section \ref{sec:social}, equilibrium analysis in Section \ref{sec:equilibrium} and sensitivity analysis in Section \ref{sec:sensitivity}.

\section{Literature}\label{sec:lit}

Our paper considers the time at which a strategic customer chooses to join a queue, with a first-come first-served order of service. A customer does not care how long he stays in the queue, but does care about the time at which he is served. Even if a customer does not care about when he is served, he may care about when he joins the queue because that affects his waiting time. An arrival-time game to a discrete stochastic queue first appeared in \cite{GH1983}, which introduced the ?/M/1 model: each customer chooses when to arrive at a single-server queue that starts operating at some known time, aiming to minimize his minimizing waiting costs. A related problem concerns the concert queuing game, where customers seek service at a first-come first-served queueing system that opens at a given time, with a customer's utility depending on when he gets service and how long he waits in the queue \cite{JS2013}. For a server with no defined starting time or ending time, customers who care only about the length of time they are in the queue will spread themselves out as much as possible \cite{LvM2004}. This model has been extended in many directions: batch service \cite{GH1987}, loss systems \cite{MC2006,HR2015}, no early arrivals \cite{HK2011}, and a network of queues \cite{HJ2015}. We find that  customers who instead care about the time at which they are served will concentrate their arrival times. 

Queueing models which consider the timing of arrivals have incorporated tardiness penalties, \cite{JS2013,H2013}, order penalties \cite{R2014}, and earliness penalties \cite{SK2016}. In \cite{RHV2016} a deterministic processor sharing system with heterogeneous due dates is considered, and the pure-strategy equilibrium arrival times are derived, explicitly for some examples, and algorithmically for the general case. The latter model resembles ours in having customers who arrive together ``helping" each achieve his due date, however arrivals do not cluster. The arrival game to a \textit{last-come-first-served} system is studied by \cite{PO2012}, who experimentally compare  several service-order policies provided in \cite{BSO2016}. A key feature in all of the above is that the equilibrium solution is given by a mixed strategy, that is, customers randomize their arrival times. Furthermore, the arrival distribution has no atoms (except for boundary cases), because arriving together with other customers is never desirable.

Work in transportation (which rarely cites the work on queues) considers a bottleneck, such as a bridge, with a commuter's costs increasing with the time he spends on the road, and incurring a cost if he arrives at his destination too early or too late (\cite{V1969} and \cite{ADL1993}). The choice of departure time in transportation is analyzed in \cite{AKN2015}, who introduce a model with masses of customers departing together. They assume a fluid population with the Greenshield's congestion dynamics (e.g. \cite{MH1984}), which are, at least from a technical perspective, close to the dynamics of processor sharing queue, in the sense that all customers travelling together are slowed as congestion increases. 

Work in economics has, like us, addressed issues of timing, but unlike us does not consider customers served in the order in which they arrive. A concern in this literature is whether agents will cluster, all taking action at the same time; another concern has been whether an agent benefits from acting before or instead after others do. Some authors show how the first entrant into an industry can earn larger profits than later entrants, a result called the pioneer's advantage \cite{RF1985}, the first-mover advantage \cite{LM1988}, and order of market entry effect \cite{L1988}. Arriving last is desirable in the advertisement board timeline game of \cite{AS2016d}, where the last advertisement posted is the most visible. 

Researchers have also considered payoffs which depend on the order of arrivals. Delay may be costly, with each player preferring that others act before him. \cite{S1974} formalized such a situation, where two animals fight over a fallen prey, with the first to give up losing, and with fighting costly for both.  The situation may be reversed,with the passage of time exogenously beneficial, and players wishing to pre-empt others. In the ``grab-the-dollar" game a player can either grab the money on the table or wait for one more period, with the pot increasing over time. Each player wants to be the first to take the money, but would rather grab a larger pot. The idea has been applied to firms' decisions about when to adopt a new technology \cite{FT1985}, and about when to enter a market \cite{G1985}. When a firm should make an irreversible investment, where the return per period decreases with the number of firms that have invested, and with the cost of investment decreasing over time, is modelled by \cite{ASD2014}. A result is that the firms may cluster, even in the absence of coordination failures, informational spillovers, or positive payoff externalities. That result resembles ours, but in a very different context. For clustering arising from coordination failures, see \cite{LP2003}. For effects of positive network externalities see \cite{MW2010}, and for informational spillovers see \cite{CG1994}, and \cite{BM2010}. 

Or players may prefer to be neither first nor last; that general situation where rewards depend on the players' ordinal rank of timing action is modeled by \cite{PS2008} and by \cite{APS2014}. We too have order matter, but not by assumption but because customers have preferences over when they are served, and the order in which a customer arrives affects when he is served.
%

\section{Model}\label{sec:model}

Consider $n$ customers, each with a single job to be processed by a single server. Serving a single job takes one unit of time. The ideal service \textit{start time} for customer $i=1,\ldots,n$, is $d_i$ (due date). If customer $i$ commences service at time $s$ then his earliness time is
\[
E_i(s)=(d_i-s)^{+}=\max\{0,d_i-s\}.
\]
His tardiness time is
\[
T_i(s)=(s-d_i)^{+}=\max\{0,s-d_i\}\ .
\]
The deviation cost incurred is
\begin{equation}\label{eq:d_cost}
DC_i(s)=\gamma E_i(s)+ \beta T_i(s), \quad i=1,\ldots,n\ ,
\end{equation} 
where $\beta$ is the marginal lateness cost and $\gamma$ is the marginal earliness cost. In the case of symmetric deviation penalties, $\gamma=\beta$, the cost function is simply the absolute deviation from the due date multiplied by a constant.

In the classical scheduling setting a central decision maker selects a sequence of service start times $(s_1,\ldots,s_n)$ that minimizes the total weighted deviation cost,
\[
\sum_{i=1}^n DC_i(s_i)=\gamma\sum_{i\in\mathcal{E}}|s_i|+\beta\sum_{i\in\mathcal{T}}s_i\ ,
\]
where $\mathcal{E}=\{i:\ E_i>0\}$ is the set of early jobs and $\mathcal{T}=\{i:\ T_i\geq 0\}$ is the set of tardy jobs. The solution of the centralized problem will be detailed in Section \ref{sec:social} and will serve as a benchmark for the following decentralized analysis.

In the decentralized setting each customer decides when to send his job to the system. We adopt the standard definition of a Nash equilibrium. That requires two conditions. First, each customer maximizes his expected utility, given his beliefs about what other customers will do, by choosing when to send his job. Second, each customer's beliefs are rational -- he does not believe that some customer will send a job at a time which that customer would avoid. Such an equilibrium can arise in at least two plausible ways. First, each customer can calculate, as we do, what are equilibrium times, and behave accordingly. Second, each customer may observe when others sent jobs in the past, or when customer sent jobs at other similar facilities, and so, having such information, decide when is the best time for himself to send a job. An equilibrium will then be self-sustaining---the times of submission are described by the equilibrium (that is what customers observed in the past or at other facilities), and each customer has an incentive to send a job at a time described by the equilibrium. Moreover, as is standard in analyses of Nash equilibria, a customer's behavior may be described by a probability that he will send a job at any stated time; customers then have beliefs about the probabilities adopted by other customers. Formally, an action of customer $i$ is an arrival time $t_i\in\mathcal{A}$, where $\mathcal{A}$ denotes the interval of time that the server is available. A mixed strategy is a probability distribution of arrival times, i.e., a random variable with a cdf $F_i$ such that $F_i(\mathcal{A})=1$.

The server admits jobs according to a FCFS regime. Multiple jobs that arrive at the same instant are admitted into service in uniform random order. Hence, for a given arrival profile $\mathbf{t}=(t_1,\ldots,t_n)$ the effective time of admittance into service depends on the arrival of others and on the lottery for order for jobs arriving together in clusters. To demonstrate the system dynamics, consider an ordered profile of arrivals: $t_1\leq t_2\leq\cdots\leq t_n$, where the order already takes into account the lottery for simultaneous arrivals. The first job is admitted immediately:
\[
s_1=t_1\ .
\] 
If the server is idle the subsequent jobs are admitted when they arrive. If the server is busy they are admitted immediately after the previous job is completed : 
\[
s_j=\max\{s_{j-1}+1,t_j\}\ 1<j\leq n \ .
\]

Let $S_i(\mathbf{t})$ denote the random variable of the start time for job $i$ given the arrival profile $\mathbf{t}$. The expected deviation cost \eqref{eq:d_cost} of customer $i$ is
\begin{equation}\label{eq:cost}
c_i(\mathbf{t})=\E\left[DC_i(S_i(\mathbf{t}))\right] \ , 
\end{equation}
where the expectation accounts for both the lottery for order among simultaneous arrivals and the randomization of arrival times when some customers used mixed strategies.

To summarize, the dynamics of the non-cooperative game are as follows:
\begin{enumerate}[label=(\alph*)]
\item Customers simultaneously decide when to send their jobs.
\item The server starts working when the first job arrives.
\item Jobs arriving at a busy server form a FCFS queue.
\item Jobs arriving simultaneously are randomly ordered at the end of the queue.
\item Customers pay a penalty for deviation of their service start time from the due date.
\end{enumerate}

\begin{remark}
We assume that the due date is for the service start time, as is common in queueing models (e.g. \cite{H2013}), and not completion time, as is common in the scheduling literature. When service times are deterministic this clearly does not affect the outcome in both the centralized and decentralized settings. When service times are stochastic the distinction needs to be treated with more care.
\end{remark}

\textbf{Base model assumptions}:
\begin{enumerate}
\item Homogeneous due date, $d_i=0$, $\forall i=1,\ldots,n$.
\item Unrestricted server availability, $\mathcal{A}=(-\infty,\infty)$.
\end{enumerate}

These assumptions are appropriate for a system that is available for a long time. In standard scheduling models the server becomes available at $t=0$ and the the due date is some positive time $d>0$. If the server is available for long enough we can normalize the due date to zero, thereby simplifying the presentation of the analysis. Section \ref{sec:sensitivity} examines the implications of these assumptions, among other extensions, and shows that the equilibrium results of the base model still hold in many cases. In particular, Section \ref{sec:restricted} considers availability restrictions to the system and Section \ref{sec:heterogeneous} examines the consequences of heterogeneous due dates.

\section{Social optimization}\label{sec:social}
A central planner seeking to minimize the total cost,
\[
TC(\mathbf{t})=\sum_{i=1}^n c_i(\mathbf{t})\ ,
\]
can obtain any possible sequence of service start times by setting arrival times at a distance of at least $1$ from each other. Therefore the social optimization problem is the classical sequencing problem,
\begin{equation*}\label{eq:social_opt}
\begin{split}
\min_{(s_1,\ldots,s_n)}\sum_{i=1}^n[\gamma E_i(s_i)+\beta T_i(s_i)]\ ,\\
\mathrm{s.t.} \quad s_i\geq s_{i-1}+1, \quad i=2,\ldots,n\ .
\end{split}
\end{equation*}

Note, however, that the identical sequence of start times can be achieved by several \textit{arrival profiles}, for example two customers arriving at the same instant will yield the same sequence as that of one arriving exactly a unit after the other.

We next show that with a common due-date the optimal service sequence is any sequence with no idle time such that $0$ is the $\frac{\beta}{\beta+\gamma}$ percentile of the sequence. In other words the proportion of early jobs is $\frac{\beta}{\beta+\gamma}$. This result arises because the objective function is the sum of absolute deviations from zero with different weights for negative and positive values. Proposition \ref{prop:homogeneous_opt} is a standard result in single machine sequencing (e.g. \cite{BS1990}). Nevertheless, we provide the statement and proof for our specific formulation for  completeness and for easy comparison to the subsequent equilibrium analysis. 

\begin{proposition}\label{prop:homogeneous_opt}
Let $\tilde{s}:=\frac{n\beta}{\beta+\gamma}$. A sequence of service start times is optimal if and only if
\[
s_i=s_{i-1}+1,\quad i=2,\ldots,n\ ,
\]
and 
\begin{enumerate}
\item if $\tilde{s}\in\mathbbm{N}$ then
\[
s_1\in \left[-\tilde{s},-\tilde{s}+1\right]\ ,
\]
\item if $\tilde{s}\notin\mathbbm{N}$ then
\[
s_1=-\floor{\tilde{s}}\ .
\]
\end{enumerate} 
\end{proposition}
\begin{proof}
Without loss of generality let $s_1\leq s_2\leq \cdots\leq s_n$. Because $d_i=0$ for all $i=1,\ldots,n$, then \eqref{eq:cost} yields
\[
c_i(s)=DC_i(s)=|s|(\beta\mathbf{1}_{\{s> 0\}}+\gamma\mathbf{1}_{\{s< 0\}})\ ,
\]
where $\mathbf{1}_A$ is the indicator function on condition $A$. An optimal service sequence has no idle time between jobs; that is, the sequence satisfies 
\[
s_i=s_{i-1}+1=s_1+i-1,\ \forall i\geq 2\ .
\]
This result holds because if the server were idle at some time, for instance after job $i$ such that $s_i<0$, then the earliness penalty of all jobs up to $i$ can be reduced by shifting them forward until there is no idle time. This shift does affect the start time of subsequent jobs. The reverse argument can be used for $s_i>0$.

Hence there is a single decision variable: $s_1$. Denoting $i_0:=\max\{i:s_i\leq 0\}$, the total cost function is
\[
\begin{split}
TC(s_1) &= \sum_{i=i_0+1}^{n}\beta(s_1+(i-1))-\sum_{i=1}^{i_0}\gamma(s_{1}+(i-1)) \\
&= \big(\beta(n-i_0)-\gamma i_0\big)s_1+\sum_{i=i_0+1}^{n}\beta(i-1)-\sum_{i=1}^{i_0}\gamma(i-1) \ .
\end{split}
\]
Note that $i_0$ depends on the choice of $s_1$. Therefore, the cost function is piecewise affine with respect to $s_{1}$, with the slope determined by $i_0$: $\beta(n-i_0)-\gamma i_0$. The slope is negative if $i_0>n\frac{\beta}{\beta+\gamma}$, positive if $i_0<n\frac{\beta}{\beta+\gamma}$, and zero if $i_0=n\frac{\beta}{\beta+\gamma}$ (in case this is indeed an integer). Therefore, if $\frac{n\beta}{\beta+\gamma}\in\mathbbm{N}$ then any $s_1$ satisfying 
\[
i_0=\max\{i:s_1+i-1\leq 0\}=\frac{n\beta}{\beta+\gamma}\ ,
\]
or equivalently,
\[
-\frac{n\beta}{\beta+\gamma}\leq s_1\leq -\frac{n\beta}{\beta+\gamma}+1 \ ,
\]
is optimal. Otherwise, if $\frac{n\beta}{\beta+\gamma}\notin\mathbbm{N}$ then the objective function has a global optimum that satisfies
\[
\frac{n\beta}{\beta+\gamma}<i_0=\max\{i:s_1+i-1\leq 0\}<\frac{n\beta}{\beta+\gamma}+1\  ,
\]
and $s_1+i_0-1=0$. Hence, we conclude that $i_0=\ceil{\frac{n\beta}{\beta+\gamma}}$, and lastly, that
\[
s_1=-\ceil{\frac{n\beta}{\beta+\gamma}}+1=-\floor{\frac{n\beta}{\beta+\gamma}}-1+1=\floor{\frac{n\beta}{\beta+\gamma}}\ .
\]
\end{proof}
\begin{corollary}
If $\beta=\gamma$ then a sequence of service start times $s_i=s_1+i-1$ is optimal if and only if zero is a median of the sequence.
\end{corollary}

\section{Equilibrium analysis}\label{sec:equilibrium}

An arrival profile $\mathbf{t}=(t_1,\ldots,t_n)$ is a pure-strategy Nash equilibrium if
\[
c_i(\mathbf{t})\leq c_i(t_1,\ldots,t_i^{'},\ldots,t_n),\quad \forall t_i^{'}\in\mathcal{A},\ i=1,\ldots,n\ .
\]
This means that the arrival profile is stable in the sense that no single customer can reduce his expected deviation penalty by sending his job at a different time.

\subsection{Two-customer game}\label{sec:eq_n2}

Consider first a two-customer example with symmetric deviation penalties: $n=2$ and $\beta=\gamma$. In this case $t_1=t_2=-\frac{1}{2}$ is the unique pure strategy equilibrium with both customers incurring a cost of $\frac{\beta}{2}$. No customer has an incentive to deviate because arriving during $\big(-\frac{1}{2},\frac{1}{2}\big]$ will cost exactly $\frac{\beta}{2}$; arriving later or earlier than this interval is clearly more costly. The uniqueness can be derived by considering the best-response function to any arrival $t$ by the other customer (we use the notation $t-$ to indicate that arriving momentarily before the other customer is optimal):
\[
b(t)=\left\lbrace\begin{array}{ll}
0 \mbox{,}& t>0, \\
t- \mbox{,}& -\frac{1}{2}< t\leq 0, \\
\left[-\frac{1}{2},\frac{1}{2}\right] \mbox{,}& t=-\frac{1}{2}, \\
\left(t,1+t\right] \mbox{,}& -1\leq t< -\frac{1}{2} \\
0 \mbox{,}& t< -1.
\end{array}\right.
\]
The best-response is not necessarily unique. The pair of arrival times $(t_1,t_2)$ is an equilibrium if $t_1\in b(t_2)$ and $t_2\in b(t_1)$. Indeed, $t_1=t_2=-\frac{1}{2}$ is the only pair that satisfies this condition. This is illustrated for symmetric strategies in Figure \ref{fig:N2_BR}. We conclude that the only stable pair is $-\frac{1}{2}$ for both customers. Furthermore, the service start times in equilibrium are socially optimal. 

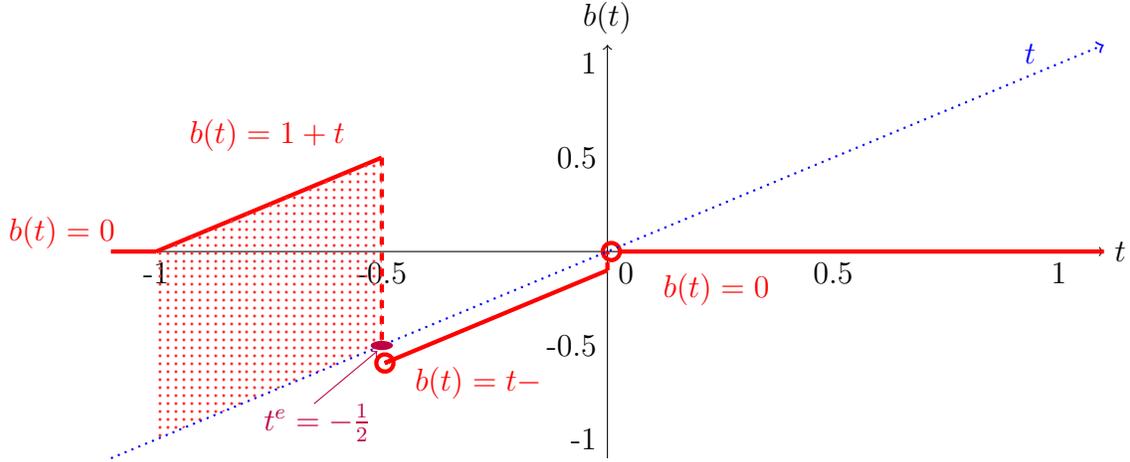
\begin{figure}[H]
\centering
\begin{tikzpicture}[xscale=6,yscale=2.5]
  \def\xmin{-1.1}
  \def\xmax{1.1}
  \def\ymin{-1.1}
  \def\ymax{1.1}
    \draw[->] (\xmin,0) -- (\xmax,0) node[right] {$t$};
    \draw[->] (0,\ymin) -- (0,\ymax) node[above] {$b(t)$};
    \foreach \x in {-1,-0.5,0.5,1}
    \node at (\x,0) [below] {\x};
    \draw[] (0,0) node[below right] {0};
    \foreach \y in {-1,-0.5,0.5,1}
    \node at (0,\y) [left] {\y};

	\draw[smooth,red, ultra thick]  (\xmin,0) -- (-1,0);	
    \draw[-,red,domain=-1:-0.5,ultra thick]  plot (\x, {1+\x});
    \draw[dashed,red,ultra thick]  (-0.5,0.5) -- (-0.5,-0.5);  
    \draw[o-,red,domain=-0.515:0,ultra thick]  plot (\x, {\x-0.1});
    \draw[dashed,red, ultra thick]  (0,-0.1) -- (0,0);
    \draw[o-,smooth,red, ultra thick]  (-0.015,0) -- (\xmax,0);     
    
    \draw[] (-1.35,0.1) node[right,red] {$b(t)=0$};
    \draw[] (-0.95,0.62) node[right,red] {$b(t)=1+t$};
    \draw[] (-0.45,-0.7) node[right,red] {$b(t)=t-$};
    \draw[] (0.1,-0.2) node[right,red] {$b(t)=0$};
    
    \fill[pattern color=red!80, opacity=1,pattern=dots] (-1,0)--	(-0.5,0.5)-- (-0.5,-0.5)-- (-1,-1);
    
    \draw[->,blue,dotted, thick,domain=\xmin:\xmax]  plot (\x, {\x});
	\draw[] (0.9,1.05) node[right,blue] {$t$};
	
	\fill[purple] (-0.5,-0.5) circle[radius=0.7pt];
	\draw[] (-0.5,-0.91) node[left,purple] {$t^e=-\frac{1}{2}$};
	\draw[->,purple,smooth] (-0.65,-0.81) -- (-0.51,-0.53);	
\end{tikzpicture}
\caption{Example: $n=2$ and $\beta=\gamma$. Best response to the arrival time $t$ of the second customer. The only fixed point is the discontinuity point at $t=-\frac{1}{2}$. The solid red line and the dotted region is the best response function $b(t)$.}\label{fig:N2_BR}
\end{figure}

\subsection{General equilibrium properties}\label{sec:eq_char}

Consider next the equilibrium for any finite number of homogeneous customers and any $\beta,\gamma>0$. It turns out that the set of pure-strategy equilibria includes only symmetric equilibria, and is given by an interval of arrival times $\tau=[\underline{t},\overline{t}]$. That is, for any $t\in\tau$ all customers simultaneously arriving at time $t$ is an equilibrium. Before stating the main result we prove several useful lemmas. All proofs are given in the Appendix. 

\begin{lemma}\label{lemma:n_equi}
Any equilibrium arrival profile satisfies the following:
\begin{enumerate}
\item[(a)] The first customer arrives at some $t_a<0$ and the last customer enters service at some $t_b>0$.
\item[(b)] The server operates continuously during the interval $[t_a,t_b]$.
\end{enumerate}
\end{lemma}

\begin{lemma}\label{lemma:no_asymmetric}
There is no asymmetric pure or mixed strategy equilibrium.
\end{lemma}

\begin{lemma}\label{lemma:no_mixed}
There is no symmetric mixed-strategy equilibrium.
\end{lemma}

The conclusion of this section is that any Nash equilibrium is pure and symmetric, that is, it is given by a single arrival time for all customers. In the following we will characterize all such equilibria and explore when the socially optimal solution is also an equilibrium. 

\subsection{Symmetric equilibrium}\label{sec:eq_solution}

If all customers arrive at time $t<0$ then the expected cost for each of them is determined by the uniform random ordering, resulting in a cost of
\[
\sum_{i=0}^{n-1}\frac{t+i}{n}\left[\beta\mathbf{1}_{\{t+i\geq 0\}}-\gamma\mathbf{1}_{\{t+i<0\}}\right].
\]
If properties (a) and (b) of Lemma \ref{lemma:n_equi} are satisfied, then a single customer has two reasonable options to deviate: arriving before everyone else and incurring a cost of at least $-\gamma t$, or obtaining service after all others have completed their service with a cost of at least $\beta(t+n-1)$. The latter cost can be achieved by arriving at any time during the interval $(t,t+n-1]$. Hence, the expected cost of choosing time $s$ when all others play $t$ is 
\begin{equation}\label{eq:cost_ts}
c(s;t)=\left\lbrace\begin{array}{ll}
-\gamma s \mbox{,}& s<t, \\
\frac{1}{n}\left[\beta\sum_{i=i_t+1}^{n-1} (t+i)-\gamma\sum_{i=0}^{i_t}(t+i)\right] \mbox{,}& s=t, \\
\beta\max\{t+n-1,s\} \mbox{,}& s>t,
\end{array}\right.
\end{equation}
where $i_t:=\max\{i:t+i<0\}$.

For $t$ to be a symmetric equilibrium the costs associated with deviating from $t$ must exceed the expected cost associated with arriving at $t$ with all the others. An immediate symmetric equilibrium $t^e$ is given by
\[
-\gamma t^e=\beta(t^e+n-1) \Leftrightarrow t^e=-(n-1)\frac{\beta}{\beta+\gamma}\ .
\]
This is an equilibrium because from \eqref{eq:cost_ts} we have that
\[
\begin{split}
c(t^e;t^e) &=  -\frac{1}{n}\gamma t^e+\frac{1}{n}\left[\beta\sum_{i=i_{t^e}+1}^{n-2} (t^e+i)-\gamma\sum_{i=1}^{i_{t^e}}(t^e+i)\right]+\frac{1}{n}\beta(t^e+n-1) \\
&\leq \beta(t^e+n-1)=-\gamma t^e\ ,
\end{split}
\]
as every element in the interior sums is smaller than the first and last elements (the start times are closer to zero), and therefore neither arriving before nor after all others benefits any one customer.

The above equilibrium is unique for $n=2$ but not in the general case. For example, suppose that all arrive at $t^e+\epsilon$ for some small $\epsilon>0$. If $n>2$ then by continuity of the cost function both extremal costs exceed the expected cost obtained by arriving at $t^e$. This behaviour is illustrated for an example with $n=5$ in Figure \ref{fig:Ctt}.

\begin{figure}[H]
\centering
\begin{tikzpicture}[xscale=2.5,yscale=0.8]
  \def\xmin{-4}
  \def\xmax{0}
  \def\ymin{0}
  \def\ymax{4.1}
    \draw[->] (\xmin,\ymin) -- (\xmax,\ymin) node[right] {$t$};
    \draw[->] (\xmin,\ymin) -- (\xmin,\ymax) node[above] {};
    \foreach \x in {-4,-3,-1,0}
    \node at (\x,\ymin) [below] {\x};
    \foreach \y in {0,1,2,3,4}
    \node at (\xmin,\y) [left] {\y};

    \draw[-,red,densely dotted,thick,domain=(\xmin:\xmax)]  plot (\x, {-\x});
    \draw[-,blue,densely dotted,thick,domain=(\xmin:\xmax)]  plot (\x, {\x+4});
    
    \draw[smooth] (-4,2)--	(-3.9,1.94)--	(-3.8,1.88)--	(-3.7,1.82)--	(-3.6,1.76)--	(-3.5,1.7)--	(-3.4,1.64)--	(-3.3,1.58)--	(-3.2,1.52)--	(-3.1,1.46)--	(-3,1.4)--	(-2.9,1.38)--	(-2.8,1.36)--	(-2.7,1.34)--	(-2.6,1.32)--	(-2.5,1.3)--	(-2.4,1.28)--	(-2.3,1.26)--	(-2.2,1.24)--	(-2.1,1.22)--	(-2,1.2)--	(-1.9,1.22)--	(-1.8,1.24)--	(-1.7,1.26)--	(-1.6,1.28)--	(-1.5,1.3)--	(-1.4,1.32)--	(-1.3,1.34)--	(-1.2,1.36)--	(-1.1,1.38)--	(-1,1.4)--	(-0.9,1.46)--	(-0.8,1.52)--	(-0.7,1.58)--	(-0.6,1.64)--	(-0.5,1.7)--	(-0.4,1.76)--	(-0.3,1.82)--	(-0.2,1.88)--	(-0.1,1.94)--	(0,2);
    
    \fill[color=gray!80, opacity=0.2,pattern=north west lines]   (-2.67,1.33) -- (-2.6,1.32)--	(-2.5,1.3)--	(-2.4,1.28)--	(-2.3,1.26)--	(-2.2,1.24)--	(-2.1,1.22)--	(-2,1.2)--	(-1.9,1.22)--	(-1.8,1.24)--	(-1.7,1.26)--	(-1.6,1.28)--	(-1.5,1.3)--	(-1.4,1.32) -- (-1.33,1.33) -- (-2,2) -- (-2.67,1.33);
    
    \draw[dashed,purple,thick]  (-2.67,0) -- (-2.67,1.33);
    \draw[] (-2.67,0) node[below,purple] {$\underline{t}$};
    \draw[dashed,purple,thick]  (-1.33,0) -- (-1.33,1.33);  
    \draw[] (-1.33,0) node[below,purple] {$\overline{t}$}; 
    
    \draw[] (-0.4,1.8) node[above] {$c(t;t)$}; 
    \draw[] (-0.3,0.4) node[above,red] {$-\gamma t$}; 
    \draw[] (-1.25,3) node[above,blue] {$\beta(t+n-1)$}; 
    
    \draw[dashed,black]  (-2,0) -- (-2,1.2);
    \node at (-2,\ymin) [below] {$t^e=-2$};	
\end{tikzpicture}
\caption{The cost of all customers arriving simultaneously at $t$, i.e. $c(t;t)$, compared with the cost a single customer can obtain by either deviation: $-\gamma t$ by arriving a moment before, or $\beta(t+n-1)$ by arriving after all are served. Arriving at $t$ is a best response to all others arriving at $t$ along the interval $(-2.66,-1.33)$ where both deviations are more costly. Example parameters: $n=5$, $\beta=\gamma=1$.}\label{fig:Ctt}
\end{figure}
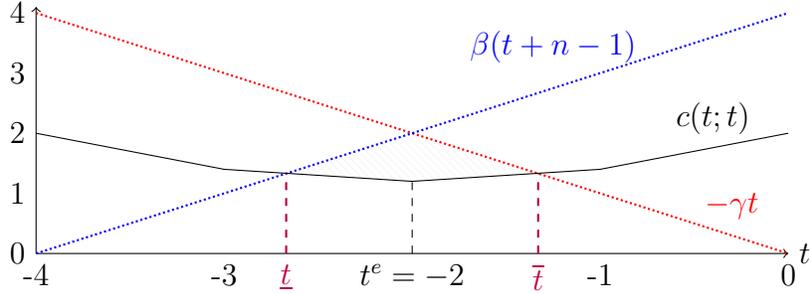

We can now fully characterize all possible equilibria. 
\begin{proposition}\label{prop:homogeneous_equil}
Let $\underline{t}$ be the unique solution of $\beta(t+n-1)=c(t;t)$; let $\overline{t}$ be the unique solution of $-\gamma t= c(t;t)$. The set of all equilibria is given by the pure and symmetric strategies of all customers arriving together at time $t$, such that
\[
t\in\tau^e=[\underline{t},\overline{t}]\subset(-(n-1),0)\ .
\]
\end{proposition}
\begin{proof}
By Lemma \ref{lemma:n_equi}a, $t$ is not an equilibrium if $t\geq 0$ or $t\leq-(n-1)$. The equation $\beta(t+n-1)=c(t;t)$ has a solution $t\in(-(n-1),0)$ because by \eqref{eq:cost_ts}, $c(0;0)<(n-1)\beta$ and $c(-(n-1);-(n-1))>0$. The cost function \eqref{eq:cost_ts} is piecewise linear with respect to $t$, with a slope of
\[
a(t)=\frac{\beta(n-1-i_t)-\gamma i_t}{n}\ ,
\]
where $i_t:=\max\{i:t+i<0\}\geq 0$. Further, $t\in(-(n-1),0)$ implies $i_t\in(0,(n-1))$ and therefore $a(t)<\beta$, which in turn implies that the solution $\underline{t}$ of $\beta(t+n-1)=c(t;t)$ is unique. Similarly, as $c(0;0)>0$, $c(-(n-1),-(n-1))<(n-1)\gamma$ and $|a(t)|<\gamma$ we conclude that $-\gamma t=c(t;t)$ admits a unique solution, $\overline{t}$. Furthermore, $\beta(t+n-1)\geq c(t;t)$ is satisfied for any $t\geq \underline{t}$, and $-\gamma t\geq c(t;t)$ is satisfied for any $t\leq \underline{t}$. We conclude that the interval $\tau$ is the set of all equilibria with pure symmetric strategies. The interval is not empty because we already saw that $t^e=-(n-1)\frac{\beta}{\beta+\gamma}$ is always an equilibrium.
\end{proof}

Observe that for $n=2$ there is a unique solution $t^e=-\frac{\beta}{\beta+\gamma}$, and in particular $t^e=-\frac{1}{2}$ for $\beta=\gamma$, as was shown in Section \ref{sec:eq_n2}. Proposition \ref{prop:homogeneous_opt} characterized the socially optimal sequence of service start times, which can be achieved, for example, by a symmetric strategy of all arriving at $\tilde{s}$, i.e., at the optimal first service time. Denote this socially-optimal strategy by $t^*:=\tilde{s}$. It turns out that the social optimum is often an equilibrium, but not always, as is illustrated by an example in Figure \ref{fig:Ctt2}. We next give necessary and sufficient conditions for the socially optimal solution to be an equilibrium, with the proof in the Appendix.

\begin{figure}[H]
\centering
\begin{tikzpicture}[xscale=6,yscale=0.35]
 \def\xmin{-4}
  \def\xmax{-2}
  \def\ymin{0}
  \def\ymax{10.1}
    \draw[->] (\xmin,\ymin) -- (\xmax,\ymin) node[right] {$t$};
    \draw[->] (\xmin,\ymin) -- (\xmin,\ymax) node[above] {};
    \foreach \x in {-3,-2}
    \node at (\x,\ymin) [below] {\x};
    \foreach \y in {0,2,4,6,8,10}
    \node at (\xmin,\y) [left] {\y};

    \draw[-,red,densely dotted,thick,domain=(\xmin:\xmax)]  plot (\x, {-\x});
    \draw[-,blue,densely dotted,thick,domain=(\xmin:-2)]  plot (\x, {5*(\x+4)});
    
    \draw[smooth] (-4,2)--	(-3.9,2.02)--	(-3.8,2.04)--	(-3.7,2.06)--	(-3.6,2.08)--	(-3.5,2.1)--	(-3.4,2.12)--	(-3.3,2.14)--	(-3.2,2.16)--	(-3.1,2.18)--	(-3,2.2)--	(-2.9,2.34)--	(-2.8,2.48)--	(-2.7,2.62)--	(-2.6,2.76)--	(-2.5,2.9)--	(-2.4,3.04)--	(-2.3,3.18)--	(-2.2,3.32)--	(-2.1,3.46)--	(-2,3.6);
    
    \fill[color=gray!80, opacity=0.2,pattern=north west lines]  (-3.583333,2.09) -- (-3.3333,3.3333) -- (-2.666667,2.69) --(-2.7,2.62) -- (-2.8,2.48) -- (-2.9,2.34) -- (-3,2.2) -- (-3.1,2.18) -- (-3.2,2.16) -- (-3.3,2.14) -- (-3.4,2.12) -- (-3.5,2.1);
    
    \draw[dashed,purple,thick]  (-3.583333,0) -- (-3.583333,2.09);
    \draw[] (-3.583333,0) node[below,purple] {$\underline{t}$};
    \draw[dashed,purple,thick]  (-2.666667,0) -- (-2.666667,2.69);  
    \draw[] (-2.666667,0) node[below,purple] {$\overline{t}$}; 
    
    \draw[] (-2,4) node[above] {$c(t;t)$}; 
    \draw[] (-2,2) node[right,red] {$-\gamma t$}; 
    \draw[] (-2.1,9.2) node[above,blue] {$\beta(t+n-1)$}; 
    
    \draw[dashed,black]  (-3.333,0) -- (-3.333,2.12);
    \node at (-3.333,\ymin) [below] {$t^e$};	
    
     \node at (-3.85,-0.4) [below,purple,rectangle,draw] {$t^*=-4$};
     \draw[->,purple,smooth] (-3.85,-0.4) -- (-3.99,1.9);	
    
\end{tikzpicture}
\caption{\footnotesize{The cost of all customers arriving simultaneously at $t$, i.e. $c(t;t)$, compared with the cost a single customer can obtain by either deviation: $-\gamma t$ by arriving a moment before, or $\beta(t+n-1)$ by arriving after all are served. Arriving at $t$ is a best response to all others arriving at $t$ along the interval $\tau^e=(-3.58,-2.67)$ where both deviations are more costly. However the minimal total cost is attained at $t^*=-4$ as $c(t;t)$ is an increasing function. Example parameters: $n=5$, $\beta=5$, $\gamma=1$.}}\label{fig:Ctt2}
\end{figure}
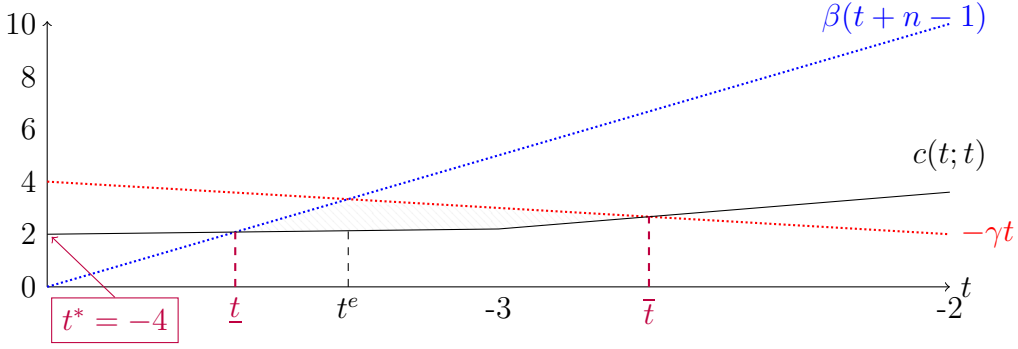

\begin{proposition}\label{prop:opt_equil}
There exists a socially optimal symmetric arrival time that is also an equilibrium if and only if $\frac{\beta}{\beta+\gamma}\in\left[\frac{1}{n},1-\frac{1}{n}\right]$.
\end{proposition}

Proposition \ref{prop:opt_equil} implies that equilibrium clustering of arrival times is often efficient. Specifically, when $n$ is not very small there is a socially optimal equilibrium for almost all values of $\beta$ and $\gamma$. This means that when the penalty function is fairly symmetric there is a socially optimal equilibrium. But when the penalty is heavily skewed in one direction the equilibrium arrivals are too early (small $\frac{\beta}{\beta+\gamma}$) or too late (large $\frac{\beta}{\beta+\gamma}$). This is illustrated in Figure \ref{fig:equil_opt}. Furthermore, note that even if for many parameter values there is a socially optimal equilibrium, welfare under most of the equilibria is less than under the socially optimal solution. In particular, the \textit{price of stability}, which is the ratio between the total costs of the best equilibrium and the social optimum, is typically one, whereas for $n>2$ the \textit{price of anarchy}, which is the ratio between the total costs of the worst equilibrium and the social optimum, exceeds $1$.

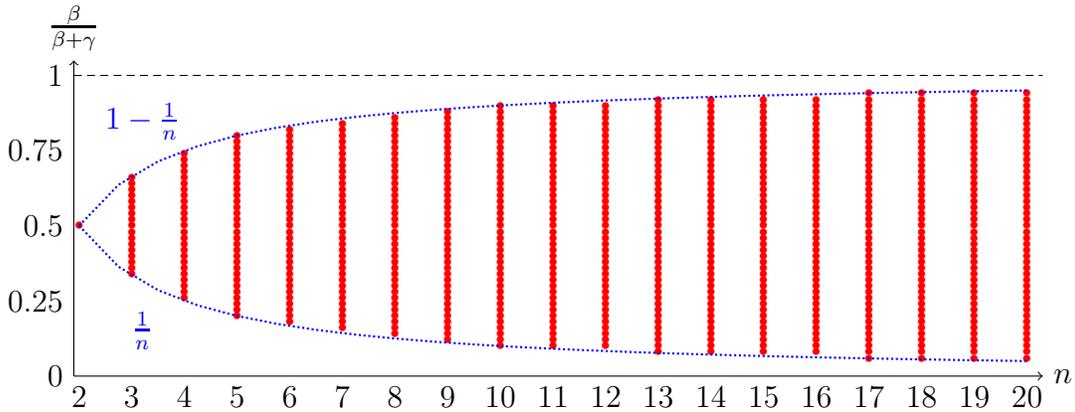
\begin{figure}[H]
\centering
\begin{tikzpicture}[xscale=0.7,yscale=4]
  \def\xmin{1.9}
  \def\xmax{20.3}
  \def\ymin{0}
  \def\ymax{1.05}
    \draw[->] (\xmin,\ymin) -- (\xmax,\ymin) node[right] {$n$};
    \draw[->] (\xmin,\ymin) -- (\xmin,\ymax) node[above] {$\frac{\beta}{\beta+\gamma}$};
    \foreach \x in {2,3,4,5,6,7,8,9,10,11,12,13,14,15,16,17,18,19,20}
    \node at (\x,\ymin) [below] {\x};
    \foreach \y in {0,0.25,0.5,0.75,1}
    \node at (\xmin,\y) [left] {\y};

	\foreach \Point in {(2,0.5),	(3,0.34),	(3,0.36),	(3,0.38),	(3,0.4),	(3,0.42),	(3,0.44),	(3,0.46),	(3,0.48),	(3,0.5),	(3,0.52),	(3,0.54),	(3,0.56),	(3,0.58),	(3,0.6),	(3,0.62),	(3,0.64),	(3,0.66),	(4,0.26),	(4,0.28),	(4,0.3),	(4,0.32),	(4,0.34),	(4,0.36),	(4,0.38),	(4,0.4),	(4,0.42),	(4,0.44),	(4,0.46),	(4,0.48),	(4,0.5),	(4,0.52),	(4,0.54),	(4,0.56),	(4,0.58),	(4,0.6),	(4,0.62),	(4,0.64),	(4,0.66),	(4,0.68),	(4,0.7),	(4,0.72),	(4,0.74),	(5,0.2),	(5,0.22),	(5,0.24),	(5,0.26),	(5,0.28),	(5,0.3),	(5,0.32),	(5,0.34),	(5,0.36),	(5,0.38),	(5,0.4),	(5,0.42),	(5,0.44),	(5,0.46),	(5,0.48),	(5,0.5),	(5,0.52),	(5,0.54),	(5,0.56),	(5,0.58),	(5,0.6),	(5,0.62),	(5,0.64),	(5,0.66),	(5,0.68),	(5,0.7),	(5,0.72),	(5,0.74),	(5,0.76),	(5,0.78),	(5,0.8),	(6,0.18),	(6,0.2),	(6,0.22),	(6,0.24),	(6,0.26),	(6,0.28),	(6,0.3),	(6,0.32),	(6,0.34),	(6,0.36),	(6,0.38),	(6,0.4),	(6,0.42),	(6,0.44),	(6,0.46),	(6,0.48),	(6,0.5),	(6,0.52),	(6,0.54),	(6,0.56),	(6,0.58),	(6,0.6),	(6,0.62),	(6,0.64),	(6,0.66),	(6,0.68),	(6,0.7),	(6,0.72),	(6,0.74),	(6,0.76),	(6,0.78),	(6,0.8),	(6,0.82),	(7,0.16),	(7,0.18),	(7,0.2),	(7,0.22),	(7,0.24),	(7,0.26),	(7,0.28),	(7,0.3),	(7,0.32),	(7,0.34),	(7,0.36),	(7,0.38),	(7,0.4),	(7,0.42),	(7,0.44),	(7,0.46),	(7,0.48),	(7,0.5),	(7,0.52),	(7,0.54),	(7,0.56),	(7,0.58),	(7,0.6),	(7,0.62),	(7,0.64),	(7,0.66),	(7,0.68),	(7,0.7),	(7,0.72),	(7,0.74),	(7,0.76),	(7,0.78),	(7,0.8),	(7,0.82),	(7,0.84),	(8,0.14),	(8,0.16),	(8,0.18),	(8,0.2),	(8,0.22),	(8,0.24),	(8,0.26),	(8,0.28),	(8,0.3),	(8,0.32),	(8,0.34),	(8,0.36),	(8,0.38),	(8,0.4),	(8,0.42),	(8,0.44),	(8,0.46),	(8,0.48),	(8,0.5),	(8,0.52),	(8,0.54),	(8,0.56),	(8,0.58),	(8,0.6),	(8,0.62),	(8,0.64),	(8,0.66),	(8,0.68),	(8,0.7),	(8,0.72),	(8,0.74),	(8,0.76),	(8,0.78),	(8,0.8),	(8,0.82),	(8,0.84),	(8,0.86),	(9,0.12),	(9,0.14),	(9,0.16),	(9,0.18),	(9,0.2),	(9,0.22),	(9,0.24),	(9,0.26),	(9,0.28),	(9,0.3),	(9,0.32),	(9,0.34),	(9,0.36),	(9,0.38),	(9,0.4),	(9,0.42),	(9,0.44),	(9,0.46),	(9,0.48),	(9,0.5),	(9,0.52),	(9,0.54),	(9,0.56),	(9,0.58),	(9,0.6),	(9,0.62),	(9,0.64),	(9,0.66),	(9,0.68),	(9,0.7),	(9,0.72),	(9,0.74),	(9,0.76),	(9,0.78),	(9,0.8),	(9,0.82),	(9,0.84),	(9,0.86),	(9,0.88),	(10,0.1),	(10,0.12),	(10,0.14),	(10,0.16),	(10,0.18),	(10,0.2),	(10,0.22),	(10,0.24),	(10,0.26),	(10,0.28),	(10,0.3),	(10,0.32),	(10,0.34),	(10,0.36),	(10,0.38),	(10,0.4),	(10,0.42),	(10,0.44),	(10,0.46),	(10,0.48),	(10,0.5),	(10,0.52),	(10,0.54),	(10,0.56),	(10,0.58),	(10,0.6),	(10,0.62),	(10,0.64),	(10,0.66),	(10,0.68),	(10,0.7),	(10,0.72),	(10,0.74),	(10,0.76),	(10,0.78),	(10,0.8),	(10,0.82),	(10,0.84),	(10,0.86),	(10,0.88),	(10,0.9),	(11,0.1),	(11,0.12),	(11,0.14),	(11,0.16),	(11,0.18),	(11,0.2),	(11,0.22),	(11,0.24),	(11,0.26),	(11,0.28),	(11,0.3),	(11,0.32),	(11,0.34),	(11,0.36),	(11,0.38),	(11,0.4),	(11,0.42),	(11,0.44),	(11,0.46),	(11,0.48),	(11,0.5),	(11,0.52),	(11,0.54),	(11,0.56),	(11,0.58),	(11,0.6),	(11,0.62),	(11,0.64),	(11,0.66),	(11,0.68),	(11,0.7),	(11,0.72),	(11,0.74),	(11,0.76),	(11,0.78),	(11,0.8),	(11,0.82),	(11,0.84),	(11,0.86),	(11,0.88),	(11,0.9),	(12,0.1),	(12,0.12),	(12,0.14),	(12,0.16),	(12,0.18),	(12,0.2),	(12,0.22),	(12,0.24),	(12,0.26),	(12,0.28),	(12,0.3),	(12,0.32),	(12,0.34),	(12,0.36),	(12,0.38),	(12,0.4),	(12,0.42),	(12,0.44),	(12,0.46),	(12,0.48),	(12,0.5),	(12,0.52),	(12,0.54),	(12,0.56),	(12,0.58),	(12,0.6),	(12,0.62),	(12,0.64),	(12,0.66),	(12,0.68),	(12,0.7),	(12,0.72),	(12,0.74),	(12,0.76),	(12,0.78),	(12,0.8),	(12,0.82),	(12,0.84),	(12,0.86),	(12,0.88),	(12,0.9),	(13,0.08),	(13,0.1),	(13,0.12),	(13,0.14),	(13,0.16),	(13,0.18),	(13,0.2),	(13,0.22),	(13,0.24),	(13,0.26),	(13,0.28),	(13,0.3),	(13,0.32),	(13,0.34),	(13,0.36),	(13,0.38),	(13,0.4),	(13,0.42),	(13,0.44),	(13,0.46),	(13,0.48),	(13,0.5),	(13,0.52),	(13,0.54),	(13,0.56),	(13,0.58),	(13,0.6),	(13,0.62),	(13,0.64),	(13,0.66),	(13,0.68),	(13,0.7),	(13,0.72),	(13,0.74),	(13,0.76),	(13,0.78),	(13,0.8),	(13,0.82),	(13,0.84),	(13,0.86),	(13,0.88),	(13,0.9),	(13,0.92),	(14,0.08),	(14,0.1),	(14,0.12),	(14,0.14),	(14,0.16),	(14,0.18),	(14,0.2),	(14,0.22),	(14,0.24),	(14,0.26),	(14,0.28),	(14,0.3),	(14,0.32),	(14,0.34),	(14,0.36),	(14,0.38),	(14,0.4),	(14,0.42),	(14,0.44),	(14,0.46),	(14,0.48),	(14,0.5),	(14,0.52),	(14,0.54),	(14,0.56),	(14,0.58),	(14,0.6),	(14,0.62),	(14,0.64),	(14,0.66),	(14,0.68),	(14,0.7),	(14,0.72),	(14,0.74),	(14,0.76),	(14,0.78),	(14,0.8),	(14,0.82),	(14,0.84),	(14,0.86),	(14,0.88),	(14,0.9),	(14,0.92),	(15,0.08),	(15,0.1),	(15,0.12),	(15,0.14),	(15,0.16),	(15,0.18),	(15,0.2),	(15,0.22),	(15,0.24),	(15,0.26),	(15,0.28),	(15,0.3),	(15,0.32),	(15,0.34),	(15,0.36),	(15,0.38),	(15,0.4),	(15,0.42),	(15,0.44),	(15,0.46),	(15,0.48),	(15,0.5),	(15,0.52),	(15,0.54),	(15,0.56),	(15,0.58),	(15,0.6),	(15,0.62),	(15,0.64),	(15,0.66),	(15,0.68),	(15,0.7),	(15,0.72),	(15,0.74),	(15,0.76),	(15,0.78),	(15,0.8),	(15,0.82),	(15,0.84),	(15,0.86),	(15,0.88),	(15,0.9),	(15,0.92),	(16,0.08),	(16,0.1),	(16,0.12),	(16,0.14),	(16,0.16),	(16,0.18),	(16,0.2),	(16,0.22),	(16,0.24),	(16,0.26),	(16,0.28),	(16,0.3),	(16,0.32),	(16,0.34),	(16,0.36),	(16,0.38),	(16,0.4),	(16,0.42),	(16,0.44),	(16,0.46),	(16,0.48),	(16,0.5),	(16,0.52),	(16,0.54),	(16,0.56),	(16,0.58),	(16,0.6),	(16,0.62),	(16,0.64),	(16,0.66),	(16,0.68),	(16,0.7),	(16,0.72),	(16,0.74),	(16,0.76),	(16,0.78),	(16,0.8),	(16,0.82),	(16,0.84),	(16,0.86),	(16,0.88),	(16,0.9),	(16,0.92),	(17,0.06),	(17,0.08),	(17,0.1),	(17,0.12),	(17,0.14),	(17,0.16),	(17,0.18),	(17,0.2),	(17,0.22),	(17,0.24),	(17,0.26),	(17,0.28),	(17,0.3),	(17,0.32),	(17,0.34),	(17,0.36),	(17,0.38),	(17,0.4),	(17,0.42),	(17,0.44),	(17,0.46),	(17,0.48),	(17,0.5),	(17,0.52),	(17,0.54),	(17,0.56),	(17,0.58),	(17,0.6),	(17,0.62),	(17,0.64),	(17,0.66),	(17,0.68),	(17,0.7),	(17,0.72),	(17,0.74),	(17,0.76),	(17,0.78),	(17,0.8),	(17,0.82),	(17,0.84),	(17,0.86),	(17,0.88),	(17,0.9),	(17,0.92),	(17,0.94),	(18,0.06),	(18,0.08),	(18,0.1),	(18,0.12),	(18,0.14),	(18,0.16),	(18,0.18),	(18,0.2),	(18,0.22),	(18,0.24),	(18,0.26),	(18,0.28),	(18,0.3),	(18,0.32),	(18,0.34),	(18,0.36),	(18,0.38),	(18,0.4),	(18,0.42),	(18,0.44),	(18,0.46),	(18,0.48),	(18,0.5),	(18,0.52),	(18,0.54),	(18,0.56),	(18,0.58),	(18,0.6),	(18,0.62),	(18,0.64),	(18,0.66),	(18,0.68),	(18,0.7),	(18,0.72),	(18,0.74),	(18,0.76),	(18,0.78),	(18,0.8),	(18,0.82),	(18,0.84),	(18,0.86),	(18,0.88),	(18,0.9),	(18,0.92),	(18,0.94),	(19,0.06),	(19,0.08),	(19,0.1),	(19,0.12),	(19,0.14),	(19,0.16),	(19,0.18),	(19,0.2),	(19,0.22),	(19,0.24),	(19,0.26),	(19,0.28),	(19,0.3),	(19,0.32),	(19,0.34),	(19,0.36),	(19,0.38),	(19,0.4),	(19,0.42),	(19,0.44),	(19,0.46),	(19,0.48),	(19,0.5),	(19,0.52),	(19,0.54),	(19,0.56),	(19,0.58),	(19,0.6),	(19,0.62),	(19,0.64),	(19,0.66),	(19,0.68),	(19,0.7),	(19,0.72),	(19,0.74),	(19,0.76),	(19,0.78),	(19,0.8),	(19,0.82),	(19,0.84),	(19,0.86),	(19,0.88),	(19,0.9),	(19,0.92),	(19,0.94),	(20,0.06),	(20,0.08),	(20,0.1),	(20,0.12),	(20,0.14),	(20,0.16),	(20,0.18),	(20,0.2),	(20,0.22),	(20,0.24),	(20,0.26),	(20,0.28),	(20,0.3),	(20,0.32),	(20,0.34),	(20,0.36),	(20,0.38),	(20,0.4),	(20,0.42),	(20,0.44),	(20,0.46),	(20,0.48),	(20,0.5),	(20,0.52),	(20,0.54),	(20,0.56),	(20,0.58),	(20,0.6),	(20,0.62),	(20,0.64),	(20,0.66),	(20,0.68),	(20,0.7),	(20,0.72),	(20,0.74),	(20,0.76),	(20,0.78),	(20,0.8),	(20,0.82),	(20,0.84),	(20,0.86),	(20,0.88),	(20,0.9),	(20,0.92),	(20,0.94)}
    {\node[red] at \Point {\tiny{\textbullet}};}
    
    \draw[densely dashed] (\xmin,1)--	(\xmax,1);
    
    \draw[-,blue,densely dotted,thick,domain=(2:20)]  plot (\x, {(\x-1)/\x});
    \node at (3.2,0.75) [above,blue] {$1-\frac{1}{n}$};
    
    \draw[-,blue,densely dotted,thick,domain=(2:20)]  plot (\x, {1-(\x-1)/\x});
    \node at (3.2,0.25) [below,blue] {$\frac{1}{n}$};
 
\end{tikzpicture}
\caption{Range of parameter values (red dotted lines) such that there is a socially optimal equilibrium.}\label{fig:equil_opt}
\end{figure}

\section{Sensitivity analysis}\label{sec:sensitivity}
It is interesting to explore what assumptions are crucial for the existence of the pure and symmetric equilibria. We consider some possible deviations from the assumptions in the basic model.

\subsection{Restricted server availability}\label{sec:restricted}

Suppose that the sever only admits jobs into the queue during the interval $\mathcal{A}=[a,b]$. We assume that all jobs in the queue at time $b$ are served so the server may still operate after $b$ even though no new jobs are allowed to join. 

Recall that by Proposition \ref{prop:homogeneous_equil} the set of pure strategy equilibria satisfies $\tau^e\subset\left(-\frac{1}{n-1},0\right)$. Therefore, if $a\leq -\frac{1}{n-1}$ and $b\geq 0$ then the equilibrium outcome is the same as in the base model. In fact, this identity holds if $a\leq \underline{t}$ and $b\geq \overline{t}$, where the thresholds are as defined in Proposition \ref{prop:homogeneous_equil}. The more interesting cases occur when the interval of equilibrium arrival times in the base model  is partially or fully unavailable. In the next proposition we characterize all possible equilibria for the restricted server availability model.

\begin{proposition}\label{prop:restricted}
Let $\tau^r=[a,b]\cap\tau^e$, where $\tau^e=[\underline{t},\overline{t}]$ as defined in Proposition \ref{prop:homogeneous_equil}.
\begin{enumerate}[label=(\alph*)] 
\item If $\tau^r\neq\emptyset$ then $\tau^r$ is the set of all pure-strategy equilibria.
\item If $\tau^r=\emptyset$ and $a> \overline{t}$ then all arriving at $t=a$ is the unique pure-strategy equilibrium.
\item If $\tau^r=\emptyset$ and $b< \underline{t}$ then all arriving at $t=b$ is the unique pure-strategy equilibrium.
\end{enumerate}
\end{proposition}

\begin{remark}
Proposition \ref{prop:restricted} shows that the restricted model may have different equilibrium outcomes than the base model. For example, if $b<\underline{t}$ then no arrivals are allowed during the interval of base model equilibria and $\tau^r=\emptyset$. Part (c) of the proposition implies that $t=b\notin\tau^e$ is the unique equilibrium. However, when any segment of the base model equilibrium interval is available, i.e., $\tau^r\neq\emptyset$, then any equilibrium in the restricted model is also an equilibrium in the base model.
\end{remark}

\subsection{Waiting cost}\label{sec:waiting}
We define the waiting time insensitivity cost as follows: suppose that there is a penalty of $\alpha>0$ per unit of waiting time in the queue until service commences. We claim that if $n>2$ there is still an interval of pure symmetric equilibrium arrival times for $\alpha$ sufficiently small.

Consider the homogeneous two-customer example: $n=2$ and $\beta=\gamma$. If both customers arrive at $t=-\frac{1}{2}$ then each individual's cost is
\[
\frac{1}{2}\left(\alpha\cdot 0+\beta\cdot\frac{1}{2}\right)+\frac{1}{2}\left(\alpha\cdot 1+\beta\cdot\frac{1}{2}\right)=\frac{1}{2}(\beta+\alpha).
\]
This is not an equilibrium because arriving $\epsilon<\alpha$ before $-\frac{1}{2}$ will reduce the cost to $\frac{1}{2}(\beta+\epsilon)$. Therefore that equilibrium is possible only with mixed strategies that are given by continuous distributions of arrival times (i.e. no atoms).

If, however, $n>2$ then $t^e=-(n-1)\frac{\beta}{\beta+\gamma}$ is still an equilibrium for sufficiently small $\alpha$. Using the notation of Section \ref{sec:equilibrium}, the cost to an individual customer when all arrive at time $t$ is
\[
c_\alpha(t;t)=c(t;t)+\frac{\alpha}{n}\sum_{i=0}^{n-1}i=c(t;t)+\alpha\frac{n-1}{2}\xrightarrow{\alpha\to 0}c(t;t) \ ,
\]
where $c(t;t)$ is given by \eqref{eq:cost_ts}. Recall that $t^e$ was the time that ensured that the first and last customers in the realised random draw incur a cost of exactly $-\gamma t^e$. Hence, if $\alpha$ is small then arriving a moment before $t^e$ or at $t^e+n-1$ yields a cost of $-\gamma t^e$, which exceeds $c(t^e;t^e)+\alpha\frac{n-1}{2}$ for small enough $\alpha$. Furthermore, by continuity there is an interval of equilibrium points $\tau^\alpha=[\underline{t},\overline{t}]$ (see Proposition \ref{prop:homogeneous_equil}), where the interval bounds depend on $\alpha$ as follows,
\begin{align*}
\underline{t} &= \inf\left\lbrace t<0:\ \beta(t+n-1)-\alpha\frac{n-1}{2}\geq c(t;t)\right\rbrace\ , \\
\overline{t} &= \sup\left\lbrace t<0:\ -\gamma t-\alpha\frac{n-1}{2}\geq c(t;t)\right\rbrace\ .
\end{align*}
In particular, the interval of equilibria shrinks with $\alpha$; that is, $\tau^{\alpha_1}\subset\tau^{\alpha_2}\subset\tau^{0}$ for any $\overline{\alpha}>\alpha_1>\alpha_2>0$, where $\overline{\alpha}$ is the highest waiting penalty that still attains a pure equilibrium (which is unique because $\underline{t}=\overline{t}$). 

\subsection{Stochastic service times}\label{sec:stochastic}

Let $G$ be the common cdf of a customer's service time. We will present two examples that show that the atomic solution of all arriving together is still an equilibrium for some service distributions, but not for all. The first example is a two-customer game with two possible service times, which can be thought of as `high' and `low', with each customer independently drawing one of them with probability $1/2$. In this case there always exists a symmetric pure strategy equilibrium. The next example has exponential service times. We show that for $n=2$ there is no pure-strategy equilibrium, but for $n>2$ there is an interval of symmetric equilibria.

\begin{example} $n=2$ and discrete uniform service times with two possible service lengths $\{a,b\}$, without loss of generality $a=1$ and $b=2$, given by
\[
G(x)=\left\lbrace\begin{array}{ll}
0 \mbox{,} & x<1, \\
\frac{1}{2} \mbox{,} & 1\leq x<2, \\
1 \mbox{,} & x\geq 2. 
\end{array}\right.
\]
The unique symmetric pure equilibrium is
\[
t^e=\left\lbrace\begin{array}{ll}
-\frac{3\beta}{2(\beta+\gamma)} \mbox{,} & \beta\leq 2\gamma,\\
-\frac{2\beta-\gamma}{\beta+\gamma} \mbox{,} & \beta> 2\gamma.
\end{array}\right.
\]
Suppose that customer $2$ arrives at $t \in (-1,0)$; then customer $1$ has four reasonable arrival times
\begin{enumerate}
\item[(a)] A moment, say $\epsilon$, before $t$, with a cost of $-\gamma t+\gamma\epsilon$.
\item[(b)] At exactly $t$, with a cost of
\[
-\frac{1}{2}\gamma t+\frac{1}{2}\left[\frac{1}{2}\beta(t+1)+\frac{1}{2}\beta(t+2)\right]\ .
\]
\item[(c)] At $t+1$, with a cost of
\[
\frac{1}{2}\beta(t+1)+\frac{1}{2}\beta(t+2)\ .
\]
\item[(d)] At $t+2$, with a cost of $\beta(t+2)$.
\end{enumerate}
Clearly, (d) is worse than (c) for any $\gamma$ and $\beta$. If $\beta \leq 2\gamma$ the value $t^e=-\frac{3\beta}{2(\beta+\gamma)}$ ensures that the cost (c) is $-\gamma t^e$, and consequentially is the same as (b). In this case deviating to (a) is sub-optimal for any $\epsilon>0$. Indeed, the best response function has the same form as the deterministic service example illustrated in Figure \ref{fig:N2_BR}.

The case of $\beta> 2\gamma$ can be verified using similar arguments.
\end{example}

\begin{example}
Consider $n$ customers with independent exponential service times, $G(x)=1-e^{-x}$. Denote the convolution of $k$ service times by $Y_k=\sum_{i=1}^k X_i$, where $X_i\sim\mathrm{Exp}(1)$, and note that $Y_k$ follows an Erlang distribution with parameters $(k,1)$. If all arrive at time $t<0$ then the expected cost for each customer is 
\begin{align*}
c_G(t;t) &=\frac{1}{n}\left[-\gamma \left(t+\sum_{i=1}^{n-1}\int_0^{-t}(t+s)\frac{s^{i-1}e^{-s}}{(i-1)!}\ ds\right)+\beta\sum_{i=1}^{n-1}\int_{-t}^{\infty}(t+s)\frac{s^{i-1}e^{-s}}{(i-1)!}\ ds\right] \\ 
&= \frac{1}{n}\Big[-\gamma t-\gamma\sum_{i=1}^{n-1}(t\P(Y_i\leq -t)+i\P(Y_{i+1}\leq -t))\\
&\mbox{}\quad \quad +\beta\sum_{i=1}^{n-1}(t\P(Y_i>-t)+i\P(Y_{i+1}>-t))\Big]\ .
\end{align*}

For any symmetric arrival time $t<0$ there are two reasonable options for a single customer to deviate: arriving momentarily before $t$ with a cost of $\approx-\gamma t$; arriving at $s=0$ which yields a cost that is at least as small as any $s>0$ for every realization of the service times and strictly smaller than the cost incurred by arriving at $s\in(t,0)$. When deviating to $s=0$ the expected cost is
\[
c_G(0;t)=\beta\int_{-t}^\infty (t+s)\frac{s^{n-2}e^{-s}}{(n-2)!}\ ds=\beta(t\P(Y_{n-1}>-t)+(n-1)\P(Y_{n}>-t))\ .
\]
Therefore, the set of pure symmetric equilibria is given by the interval $\tau=[\underline{t},\overline{t}]$, where $\underline{t}=\inf\{t<0:\ c_G(0,t)\geq c_G(t;t)\}$ and $\overline{t}=\sup\{t<0:\ -\gamma t\geq c_G(t;t)\}$.

For $n=2$ and $t<0$,
\[
c_G(t;t)=-\frac{1}{2}\gamma t+\frac{1}{2}\left[-\gamma(1+t)+(\beta+\gamma)e^t\right]\ ,
\]
and
\[
c_G(0;t)=\beta e^{t}\ .
\]
Equating $-\gamma t=c(t;t)$ yields
\[
\overline{t}=\log\left(\frac{\gamma}{\beta+\gamma}\right)\ ,
\]
i.e., if a customer arrives at $t\in\left(\log\left(\frac{\gamma}{\beta+\gamma}\right),0\right]$ then the other customer will arrive a moment earlier. However, for any $t\leq\overline{t}$ we have 
\[
c_G(0;t)=\beta e^{t}< -\gamma(1+t)+(\beta+\gamma)e^t\leq c_G(t;t)\ ,
\]
and thus $\overline{t}<\underline{t}$ and $\tau=\emptyset$. We conclude that there is no pure symmetric equilibrium in the two-customer game with exponential service times.

Numerical analysis suggests that $n=2$ is the exception and that for $n>2$ there is always an interval of pure and symmetric equilibria, as in section \ref{sec:eq_solution}. $\blacktriangle$
\end{example}

The above examples show that with a stochastic service time equilibria in pure strategies may not exist, but such scenarios need to be carefully constructed. In general, for $n>2$ multiple equilibria may exist; characterizing all of them is likely to involve elaborate analysis of multiple parameter cases.

\subsection{Random population size}\label{sec:random_pop}

Let $N$ be a random variable representing the number of arriving customers; denote the corresponding pdf by $\pi_n=\P(N=n)$. This is equivalent to stating that every arriving customer believes that there are $M$ additional customers, where
\[
q_m:=\P(M=m)=\frac{(m+1)\pi_{m+1}}{\E[N]}\ .
\]
See \cite{MM1987} for discussion and details on games with a random number of players. A special case is that of of a Poisson distributed population, which is appropriate when many customers join independently with a small probability, and then $q_m=\pi_m$.

If all customers arrive at the same time then the unconditional (on the number of total customers) probability of any customer to be the $i$-th customer in service is
\[
p_i:=\sum_{n=i-1}^\infty\frac{q_n}{n+1}=\frac{1}{\E[N]}\sum_{n=i}^\infty \pi_n=\frac{\P(N\geq i)}{\E [N]}\ ,\quad i\geq 1\ .
\]
An individual customer's cost when all arrive at time $t<0$ is then
\begin{equation}\label{eq:cost_pi}
c_\pi(t;t)=\beta\sum_{i=i_t+1}^{\infty}p_{i+1}(t+i)-\gamma\sum_{i=0}^{i_t}p_{i+1}(t+i)\ ,
\end{equation}
where, as before, $i_t=\max\{i:\ t+i-1\leq 0\}$. A customer arriving momentarily before all others incurs a cost of $-\gamma t$. If $\P(N> i_t)>0$ then the only other feasible deviation is to arrive at $t=0$, which yields a cost of
\[
c_\pi(0;t)=\beta\sum_{i=i_t+1}^{\infty}q_i(t+i)\ .
\]
As $-\gamma t$ decreases with $t$ and $c_{\pi}(0;t)$ is increasing, $-\gamma t^e=c_\pi(0;t^e)$ has a unique solution. Furthermore, by continuity there is an interval $\tau^\pi=[\underline{t},\overline{t}]$ of pure symmetric equilibria, where $c_\pi(\underline{t};\underline{t})=c_\pi(0;\underline{t})$ and $c_\pi(\overline{t};\overline{t})=-\gamma\overline{t}$.

The socially optimal symmetric strategy is given by minimizing $c_\pi(t;t)$. Taking the derivative of \eqref{eq:cost_pi},
\[
\frac{d}{dt}c_\pi(t;t)=\beta\sum_{i=i_t+1}^{\infty}p_i-\gamma\sum_{i=0}^{i_t}p_i\ ,
\]
we see that $c_\pi(t;t)$ is piecewise affine and unimodal with the minimizing $t^*$, which is possibly not unique, satisfying
\[
\beta\sum_{i=i_{t^*}+1}^{\infty}p_i\leq\gamma\sum_{i=0}^{i_{t^*}}p_i\ ,
\]
and
\[
\beta\sum_{i=i_{t^*}+2}^{\infty}p_i>\gamma\sum_{i=0}^{i_{t^*}+1}p_i\ .
\]

If $N\sim$Poisson$(\lambda)$, a Poisson sized population with mean $\lambda$, then the socially optimal solution is often an equilibrium. Figure \ref{fig:equil_opt_poisson} illustrates the range of parameters, in terms of $\frac{\beta}{\beta+\gamma}$, for which the equivalence holds. When the population is small the social optimum is only an equilibrium when the tardiness penalty is relatively small. This contrasts to the deterministic case where it only holds for $\beta$ and $\gamma$ that are close (see Proposition \ref{prop:opt_equil} and Figure \ref{fig:equil_opt}). However, as $\lambda$ increases this holds for all parameter values, as was the case for the deterministic population. Interestingly, the minimal $\frac{\beta}{\beta+\gamma}$ such that the social optimum is an equilibrium behaves approximately as $\approx e^{-\lambda}$ (the solid blue line in Figure \ref{fig:equil_opt_poisson}).

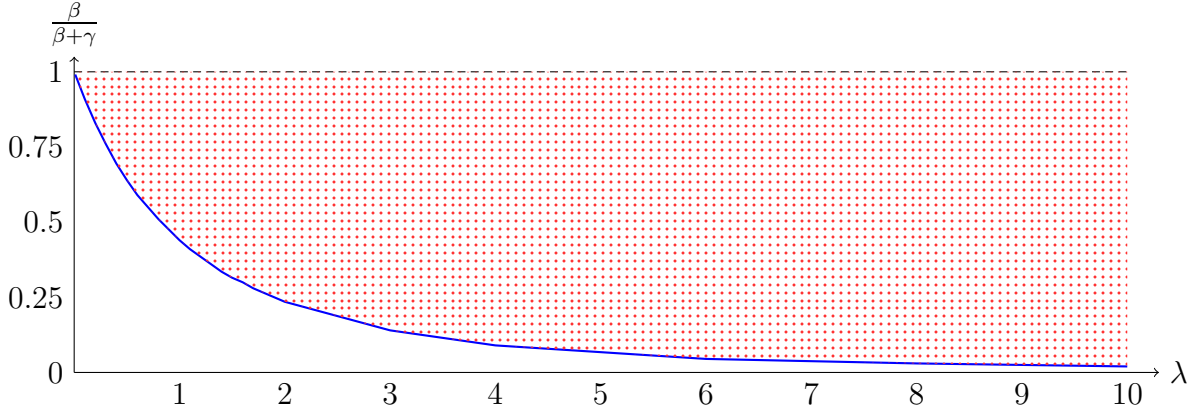
\begin{figure}[H]
\centering
\begin{tikzpicture}[xscale=1.4,yscale=4]
  \def\xmin{0}
  \def\xmax{10.3}
  \def\ymin{0}
  \def\ymax{1.05}
    \draw[->] (\xmin,\ymin) -- (\xmax,\ymin) node[right] {$\lambda$};
    \draw[->] (\xmin,\ymin) -- (\xmin,\ymax) node[above] {$\frac{\beta}{\beta+\gamma}$};
    \foreach \x in {1,2,3,4,5,6,7,8,9,10}
    \node at (\x,\ymin) [below] {\x};
    \foreach \y in {0,0.25,0.5,0.75,1}
    \node at (\xmin,\y) [left] {\y};
    
    \draw[densely dashed] (\xmin,1)--	(10,1);
    
     \draw[smooth,thick,blue] (0.01,0.991)--	(0.1,0.91)--	(0.2,0.83)--	(0.3,0.76)--	(0.4,0.695)--	(0.5,0.64)--	(0.6,0.59)--	(0.7,0.55)--	(0.8,0.51)--	(0.9,0.475)--	(1,0.44)--	(1.1,0.41)--	(1.2,0.385)--	(1.3,0.36)--	(1.4,0.335)--	(1.5,0.315)--	(1.6,0.3)--	(1.7,0.28)--	(1.8,0.265)--	(1.9,0.25)--	(2,0.235)--		(3,0.14)--		(4,0.09)--	(6,0.045)--	(8,0.03)--(10,0.02);
     
     \fill[pattern color=red!80, opacity=1,pattern=dots]  (0.01,0.991)--	(0.1,0.91)--	(0.2,0.83)--	(0.3,0.76)--	(0.4,0.695)--	(0.5,0.64)--	(0.6,0.59)--	(0.7,0.55)--	(0.8,0.51)--	(0.9,0.475)--	(1,0.44)--	(1.1,0.41)--	(1.2,0.385)--	(1.3,0.36)--	(1.4,0.335)--	(1.5,0.315)--	(1.6,0.3)--	(1.7,0.28)--	(1.8,0.265)--	(1.9,0.25)--	(2,0.235)--		(3,0.14)--		(4,0.09)--	(6,0.045)--	(8,0.03)--(10,0.02)-- (10,1)-- (0,1);
 
\end{tikzpicture}
\caption{Range of parameter values such that there is a socially optimal equilibrium. The population size is a Poisson random variable with mean $\lambda$.}\label{fig:equil_opt_poisson}
\end{figure}

\subsection{Heterogeneous customers and non-linear cost functions}\label{sec:heterogeneous}

Interestingly, it turns out that pure-symmetric equilibria exist even if the jobs have heterogeneous due dates and cost functions. However, the cost functions may not be ``too'' heterogeneous, in the sense that the due dates must be ``close" to each other, where the exact definition of ``close" depends on the population size and specific cost functions. 

\begin{assumption}\label{assump:ci}
For every customer $i=1,\ldots,n$ the cost function $c_i$ satisfies the following: 
\begin{enumerate}[label=(\alph*)]
\item There is a due date $d_i$ such that $c_i(d_i)=\min_s c_i(s)$.
\item $c_i(s)$ is continuous, unimodal, decreasing when $s<d_i$ and increasing when $s>d_i$. 
\end{enumerate}
\end{assumption}

We focus on cost functions satisfying Assumption \ref{assump:ci}, which is a reasonable and non-restrictive assumption for most scheduling and queueing scenarios. Without loss of generality we further assume that $c_i(d_i)=0$ for all $i=1,\ldots,n$. For example, the asymmetric absolute deviation cost function in \eqref{eq:cost}, and the quadratic cost function $c_i(s)=(s-d_i)^2$ both satisfy Assumption \ref{assump:ci}. As before, let $c_i(t;t')$ denote the expected cost for customer $i$ arriving at $t$ when all others arrive at $t'$.

\begin{lemma}\label{lemma:tau_i}
If Assumption \ref{assump:ci} holds then for every customer $i=1\ldots,n$ there exists a non-empty interval $\tau_i=[\underline{t}_i,\overline{t}_i]$ such that for any $t\in \tau_i$,
\[
c_i(t;t)\leq c_i(t)\ , 
\]
and
\[
c_i(t;t)\leq c_i(t+n-1)\ ,
\]
where the interval bounds are given by the solutions of
\begin{equation}\label{eq:ti_down}
c_i(\underline{t}_i+n-1)=\frac{1}{n-1}\sum_{j=0}^{n-2} c_i(\underline{t}_i+j)\ ,
\end{equation}
and
\begin{equation}\label{eq:ti_up}
 c_i(\overline{t}_i)=\frac{1}{n-1}\sum_{j=1}^{n-1} c_i(\overline{t}_i+j)\ .
\end{equation}
\end{lemma}

Lemma \ref{lemma:tau_i} characterizes the interval of symmetric arrival times $\tau_i$ such that customer $i$ does not want to deviate from the common arrival time. A symmetric equilibrium is then a time $t$ such that $t\in\tau_i$ for all $i=1,\ldots,n$. The interval of equilibria can be determined by computing all $\tau_i$ for all $i=1,\dots,n$. Alternatively, it can be computed by solving a pair of equations that we define in Proposition \ref{prop:heterogeneous_tau}. The proofs of the previous lemma and of the following proposition are in the appendix.

\begin{proposition}\label{prop:heterogeneous_tau}
The set of all pure symmetric Nash equilibria $\tau^e$ has two equivalent representations:
\begin{enumerate}[label=(\alph*)]
\item Let $\tau_i=[\underline{t}_i,\overline{t}_i]$ as defined in Lemma \ref{lemma:tau_i}, then $\tau^e=\bigcap_{i=1}^n\tau_i$.
\item $\tau^e=[\underline{t},\overline{t}]$ is given by solving 
\[
M(\underline{t})=m(\overline{t})=\frac{1}{n}\ ,
\]
where
\[
m(t)=\min_{i=1,\ldots,n}\left\lbrace\frac{c_i(t)}{\sum_{j=1}^n c_i(t+j-1)}\right\rbrace\ ,
\]
and
\[
M(t)=\min_{i=1,\ldots,n}\left\lbrace\frac{c_i(t+n-1)}{\sum_{j=1}^n c_i(t+j-1)}\right\rbrace\ .
\]
\end{enumerate}
\end{proposition}

The socially optimal arrival schedule will no longer be achieved by all arriving at the same time, because now the order of the jobs is important. The solution is given by the mathematical program
\[
\min_{s_1\in\mathbb{R},\phi\in\Pi}\sum_{i=1}^n c_i(s_1+\phi_i-1)\ ,
\]
where $\Pi$ is the space of permutations of $(1,\ldots,n)$. This formulation assumes that under the optimal schedule the server operates continuously, i.e., that the due dates are not too spread out with respect to the cost functions. Note that the social optimization problem is generally intractable and requires heuristic or approximation methods (e.g. \cite{GTW1988}).

\begin{example} Suppose that $n=3$ and that the players have a symmetric absolute deviation cost function ($\beta=\gamma=1$) and heterogeneous due dates, $(d_1,d_2,d_3)=(-0.25,0,0.25)$. Hence, the cost functions are 
\[
c_i(s)=\left\lbrace\begin{array}{ll}
|s+\frac{1}{4}|, & i=1\ , \\
|s|, & i=2\ , \\
|s-\frac{1}{4}|, & i=3\ .
\end{array}\right.
\]
To find the equilibrium use Proposition \ref{prop:heterogeneous_tau}b:
\[
M(t)=\min\left\lbrace\frac{|t+\frac{1}{4}|}{|t+\frac{1}{4}|+|t+\frac{5}{4}|+|t+\frac{9}{4}|},\frac{|t|}{|t|+|t+1|+|t+2|},\frac{|t-\frac{1}{4}|}{|t-\frac{1}{4}|+|t+\frac{3}{4}|+|t+\frac{7}{4}|}\right\rbrace\ ,
\]
and
\[
m(t)=\min\left\lbrace\frac{|t+\frac{9}{4}|}{|t+\frac{1}{4}|+|t+\frac{5}{4}|+|t+\frac{9}{4}|},\frac{|t+2|}{|t|+|t+1|+|t+2|},\frac{|t+\frac{7}{4}|}{|t-\frac{1}{4}|+|t+\frac{3}{4}|+|t+\frac{7}{4}|}\right\rbrace\ .
\]
It can be verified that $M(-1)=m(-1)=\frac{1}{3}$ and therefore all arriving at $t=-1$ is the unique symmetric pure-strategy equilibrium.
$\blacktriangle$
\end{example}

\begin{example} Suppose that $n=3$ and that due dates are heterogeneous: $(d_1,d_2,d_3)=(-0.5,0,0.5)$. We assume the following cost function, 
\[
c_i(s)=\left\lbrace\begin{array}{ll}
\beta\big(s+\frac{1}{2}\big)^2\mathbf{1}(s<-\frac{1}{2})+\big(s+\frac{1}{2}\big)^2\mathbf{1}(s>-\frac{1}{2}), & i=1\ , \\
s^2, & i=2\ , \\
\big(s-\frac{1}{2}\big)^2\mathbf{1}(s<\frac{1}{2})+\beta\big(s-\frac{1}{2}\big)^2\mathbf{1}(s>\frac{1}{2}), & i=3\ ,
\end{array}\right.
\] 
where $\beta>1$ is a parameter. The cost function can be interpreted as follows: Player 1 wants to start earlier but pays a larger penalty for early service. Player 2 wants to start at zero and has a symmetric deviation penalty. Player 3 wants to start later but has a larger penalty for starting late. 

Compute first $\tau_2=[\underline{t}_2,\overline{t}_2]$. By \eqref{eq:ti_up},
\[
\frac{1}{2}\left[(\overline{t}_2+1)^2+(\overline{t}_2+2)^2\right]=\overline{t}_2^2 \ \Rightarrow \ \overline{t}_2=-\frac{5}{6}\ ,
\]
and by \eqref{eq:ti_down},
\[
\frac{1}{2}\left[\underline{t}_2^2+(\underline{t}_2+1)^2\right]=(\underline{t}_2+2)^2 \ \Rightarrow \ \underline{t}_2=-\frac{7}{6}\ .
\]
We conclude that $\tau_2=\left[-\frac{7}{6},-\frac{5}{6}\right]$. Applying \eqref{eq:ti_up} and \eqref{eq:ti_down} yields $\overline{t}_1\geq\underline{t}_2$ if $\beta\geq 2.125$ and $\underline{t}_1\leq\overline{t}_2$ if $\beta\leq 46$. Therefore $\tau_1\cap\tau_2\neq\emptyset$ if and only if $\beta\in[2.125,46]$. Similarly for $i=3$ we find that $\overline{t}_3\geq\underline{t}_2$ if $\beta\leq 46$, $\underline{t}_3\leq\overline{t}_2$ if $\beta\geq 2.125$ and $\tau_3\cap\tau_2\neq\emptyset$ if and only if $\beta\in[2.125,46]$. 

If $\beta=5$ then $\overline{t}_1=\underline{t_3}=-1$, hence $\tau_1\cap\tau_3=-1$ and $t=-1$ is the unique equilibrium. If $\beta\in(5,46]$ then there is a non-empty (and non-singular) interval $\tau^e\subseteq\left[-\frac{7}{6},-\frac{5}{6}\right]$ of equilibrium arrival times. Lastly, if $\beta<5$ there is no pure symmetric equilibrium.
$\blacktriangle$
\end{example}

A special case of interest is considering homogeneous due dates with heterogeneous linear deviation penalties. Applying the equilibrium conditions of Proposition \ref{prop:heterogeneous_tau} yields the following characterization of the pure strategy equilibrium interval.

\begin{proposition}\label{prop:gamma_i_beta_i}
Suppose each customer $i=1,\ldots,n$ has an earliness penalty $\gamma_i$ and tardiness penalty $\beta_i$. If the deviation penalties are ordered as follows,
\[
\frac{\gamma_1}{\beta_1}\leq \frac{\gamma_2}{\beta_2}\leq\cdots\leq\frac{\gamma_n}{\beta_n}\ ,
\]
then the set of pure strategy equilibria (which may be empty) is
\[
\tau^e=[\underline{t}_n,\overline{t}_1]\ ,
\]
where $\underline{t}_n$ is given by \eqref{eq:ti_down} and $\overline{t}_1$ is given by \eqref{eq:ti_up}.
\end{proposition}

\section{Conclusion}

We analyzed a decentralized scheduling game with each customer choosing when to join a single server queue with the aim of minimizing his deviation from a due date. We find that when the due dates are close and the penalty for deviation is the significant factor in the cost function, and not the waiting cost, then in equilibrium customers arrive simultaneously. This contrasts with the literature on queue arrival timing games, specifically the ?/M/1 model of \cite{GH1983}, and the subsequent works reviewed in Section \ref{sec:lit}, where the solution is typically a mixed strategy, that is, customers randomize their arrival times. This result is only due to the unimodality around the due date of the cost function, which results in a random draw for order between the customers benefiting everyone. As shown in our sensitivity analysis, as long as the waiting costs are sufficiently small, this holds even under the standard ?/M/1 assumptions of exponential service times, Poisson distributed population size and (small) waiting costs, .

This model can be extended in several directions. Considering multiple servers, complementary to the centralized analysis surveyed in \cite{LW2004}, seems straightforward, but may be technically challenging. Another interesting extension has multiple servers combined with routing, centralized or decentralized. An additional direction is to consider customers belonging to groups with common due dates that are not close between the groups, but are not too far so there is still interaction. The question is then will groups arrive together in equilibrium but not at the same time as other groups?

\section*{Appendix - Proofs}\label{sec:appendix}

\begin{proof}[Proof of Lemma \ref{lemma:n_equi}] \mbox{}
\begin{enumerate}
\item[(a)] If the server did not so operate, then all customers incur positive expected cost, but any single customer could achieve zero cost by arriving at zero.
\item[(b)] If there is no service for some interval before zero, then the expected cost at the end of the interval is less than the expected cost at the beginning of the interval, thus contradicting the equilibrium assumption. The reverse argument holds for such an interval after time zero.
\end{enumerate}
\end{proof}

\begin{proof}[Proof of Lemma \ref{lemma:no_asymmetric}]
We separately prove the lemma for the cases of pure and mixed strategies.
\begin{enumerate}
\item Suppose that $n_x \geq 1$ and $n_y \geq 1$ customers arrive at instants $x<y$, respectively. Without loss of generality assume no one arrives before $x$ or during $(x,y)$, and that $y-x<1$ (Lemma \ref{lemma:n_equi}b). The latest possible service time of customers arriving at $x$ is $x+n_x-1$, and this is also the earliest possible service time for customers arriving at $y$. If $x+n_x-1<0$ then arriving momentarily after $x$ is better than $x$, otherwise arriving momentarily before $y$ is better than arriving at $y$. Hence, such an arrival profile cannot be an equilibrium. 
\item A mixed strategy for customer $i$ is given by a randomization of arrival times and is represented by a cdf $F_i$. If $F_i$ is a best response to the strategy of all other customers then the cost at every time of its support is constant, and at least as high at any time not in the support. If customer $i$ arrives according to the mixed strategy $F_i$, and customer $j$ arrive according to $F_j$, then there is some $t<0$ such $F_i(t)>F_j(t)$ (w.l.o.g.) and $F_i(s)=F_j(s)$ $\forall s<t$. The cost of arriving at any time $s<t$ is equal for both, but different at $t$, so if $F_i(t-)=F_j(t-)>0$ then the equilibrium assumption is contradicted. If $F_i(t-)=F_j(t-)=0$ then the support of customer $i$ starts at time $t$ and the support of $j$ starts at some time $s>t$. Customer $j$ prefers arriving at $s$ to arriving just before $t$, but this implies the same for customer $i$, contradicting the equilibrium assumption.
\end{enumerate} 
\end{proof}

\begin{proof}[Proof of Lemma \ref{lemma:no_mixed}]
A symmetric mixed strategy is given by a distribution function $F$ such that all customers randomize their arrival times according to it. Due to Lemma \ref{lemma:n_equi}a there is a pair of values $t_a$ and $t_b$, where $-(n-1)<t_a<t_b<0$, such that all customers arrive between them with probability one, i.e. $F(t_a)=0$ and $F(t_b)=1$. If $n-1$ customers arrive according to $F$ then a single customer arriving at time $t$ should incur the same expected cost at any point in the support of $F$. Using the same arguments as in the proof of Lemma \ref{lemma:no_asymmetric} we can rule out distributions with ``holes", i.e., an interval $(s_1,s_2)$ such that $F(s_1)=F(s_2)$ where $t_a\leq s_1<s_2\leq t_b$. Moreover, atoms in $F$ are not possible either, specifically $t$ such that $\lim_{\delta\to 0}F(t-\delta)\neq F(t)$. The impossibility of atoms arises because if there is an atom at $t$ then arriving just before $t$ should have the same cost as arriving at $t$ and right after it; however arriving at $t$ dominates at least one of the other options because it includes a lottery for all the customers that arrive at $t$. We are left with continuous distributions. Let $F$ be a continuous distribution with support $t\in[t_a,t_b]$, such that $t_a<0$ by Lemma \ref{lemma:n_equi}. The probability of $1\leq k\leq n-1$ arrivals during an infinitesimally small interval $[t,t+\delta)$ is approximately ${n-1 \choose k}(f(t)\delta)^k$. Hence the probability of a single arrival is approximately $(n-1)f(t)\delta+o(\delta)$ and of $k>1$ arrivals is approximately $o(\delta)$. Therefore the change of cost by arriving $\delta$ after $t_a$ instead of $t_a$ is
\[
-\gamma\delta+(n-1)f(t_a)\delta\Big(\beta(t_a+1)\mathbf{1}_{\{t_a\geq -1\}}-\gamma(t_a+1)\mathbf{1}_{\{t_a<-1\}}\Big)+o(\delta)\ ,
\]
hence dividing by $\delta$ and taking $\delta \to 0$ we get that the marginal change in cost is
\[
-\gamma+(n-1)f(t_a)\Big(\beta(t_a+1)\mathbf{1}_{\{t_a\geq -1\}}-\gamma(t_a+1)\mathbf{1}_{\{t_a<-1\}}\Big)\ ,
\]
and thus $t_a\geq -1$, as otherwise the marginal change in cost would be negative and not zero. Similarly, because $t_b>t_a\geq -1$, arriving $\delta$ before $t_b$ will result in a change of approximately
\[
-\beta f(t_b-\delta)\delta+o(\delta)\ ,
\]
which is always negative. Therefore for some sufficiently small $\delta$ the cost at $t_b-\delta$ is smaller than at $t_b$; $F$ cannot be an equilibrium.
\end{proof}

\begin{proof}[Proof of Proposition \ref{prop:opt_equil}]
In Proposition \ref{prop:homogeneous_opt} we saw that the socially optimum solution need not be unique, in particular when $\frac{n\beta}{\beta+\gamma}$ is an integer. But in any case the strategy of all arriving together at $t^*=-\floor{\frac{n\beta}{\beta+\gamma}}$ is socially optimal. We shall verify the conditions for $t^*$ to be a symmetric arrival time and show that when this is not the case there is no other socially optimal arrival time.

\begin{enumerate}
\item If $\frac{\beta}{\beta+\gamma}<\frac{1}{n}$ then $t^*=-\floor{\frac{n\beta}{\beta+\gamma}}=0$. However, $t=0$ is not an equilibrium point as long as $\beta>0$ and $\gamma>0$, because any single customer can guarantee an arbitrarily small cost by arriving momentarily before the others. Similarly, if $\frac{\beta}{\beta+\gamma}>1-\frac{1}{n}=\frac{n-1}{n}$ then $t^*=-\floor{\frac{n\beta}{\beta+\gamma}}=-(n-1)$. Here again, $t^*$ cannot be an equilibrium because a customer can arrive during $(-(n-1),0]$ and incur no cost. Furthermore, in both cases $\frac{n\beta}{\beta+\gamma}$ is not an integer and $t^*$ is the unique social optimum, and as such no equilibrium is socially optimal. 
\item If $n=2$ then the condition is only met if $\beta=\gamma$, in which case we already established that the unique equilibrium $t^e=-\frac{1}{2}$ is also socially optimal.
\item Suppose $n>2$ and $\frac{\beta}{\beta+\gamma}\in\left[\frac{1}{n},1-\frac{1}{n}\right]$. Then the social optimum is given by an integer $-(n-2)\leq t^*\leq -1$. The total cost is $nc(t;t)$ and thus minimizing the total cost is equivalent to minimizing the average cost per customer. Let $k=-t^*$, and $m=n-1-k$. Then $k\geq 1$ is the number of early arrivals and $m\geq 1$ is the number of late arrivals. Note that as $t^*$ is an integer there is always a job that starts at exactly zero and has no penalty. If all arrive at $t^*$, then the average cost of early start times is
\[
\frac{\gamma}{k}\sum_{i=1}^k i=\frac{\gamma(k+1)}{2}\ ,
\]
and the average cost of late start times, including the one at zero with no penalty, is
\[
\frac{\beta}{m+1}\sum_{i=0}^m i=\frac{\gamma m}{2}\ .
\]
Hence, a customer's expected cost is
\begin{align*}
c(t^*;t^*) &= \frac{k}{n}\cdot\frac{\gamma(k+1)}{2}+\frac{m+1}{n}\cdot\frac{\beta m}{2} \\
&= \frac{k+1}{n}\cdot\frac{\gamma k}{2}+\frac{m}{n}\cdot\frac{\beta (m+1)}{2}\ .
\end{align*}
For $t^*$ to be an equilibrium the average would have to be smaller than both of the extremal penalties $\gamma k$ and $\beta m$. By addressing all possible cases, we next prove that this property indeed holds.
\begin{enumerate}[label=\roman*.] 
\item If $\gamma k=\beta m$ then clearly the average cost is smaller than both. Each of the extremal values is smaller than the respective earliness or lateness penalty: $\gamma k>\frac{\gamma(k+1)}{2}$ and $\beta m>\frac{\beta m}{2}$. %

\item If $\gamma k>\beta m$ then arriving earlier than $t^*$ yields a cost that surely exceeds the average. Furthermore, we argue that the extremal late penalty is at least as large as the earliness average penalty: $\beta m\geq \frac{k\gamma}{2}$. Otherwise, $k\gamma>\beta 2m\geq \beta(m+1)$, but this contradicts the optimality of $t^*$: setting $t=-t^*+1$ would reduce the cost by replacing the early penalty $k\gamma$ with a smaller later penalty $\beta(m+1)$. We conclude that if $\gamma k>\beta m$ then the average cost is smaller than both extremal penalties, and thus no deviation is beneficial. 
\item If $\gamma k<\beta m$ then the arriving later than $t^*$ is clearly not a good deviation. We further argue that the earliest penalty is also at least as large as the average of lateness penalties: $\gamma k\geq \frac{\beta m}{2}$. If the inequality holds for $k=1$ (and $m=n-1$) then it holds for all values of $k$. Indeed, if $k=1$ and $\gamma < \frac{\beta (n-1)}{2}$ then $2\gamma<\beta(n-1)$ and the cost can be reduced by setting $t^*=-2$, which results in replacing the highest late penalty $\beta(n-1)$ with a smaller early penalty $2\gamma$. Lastly, we conclude that $\gamma k\geq c(t^*;t^*)$, as required.
\end{enumerate} 
\end{enumerate}
\end{proof}

\begin{proof}[Proof of Proposition
\ref{prop:restricted}] \mbox{}
Recall that all jobs arriving at $t$ is an equilibrium if no deviation to a different time is beneficial for any single customer. Specifically, arriving a moment before all others does not yield lower cost then arriving at $t$ and joining the lottery for order:
\[
-\gamma t\geq c(t;t)\ ,
\] 
and the same is true for arriving later and being served after all other jobs:
\[
\beta(t+n-1)\geq c(t;t)\ .
\] 
The first condition is met for any $t\leq\overline{t}$, the second condition is met for any $t\geq\underline{t}$, and $\tau^e=[\underline{t},\overline{t}]$ is the interval of all possible equilibria in the unrestricted model. The arguments presented here can be assisted by considering the illustration of the cost terms in Figure \ref{fig:Ctt}.
\begin{enumerate}[label=(\alph*)] 
\item If $\tau^r=[a,b]\cap\tau^e\neq\emptyset$ then $t\in\tau^e$ for any $t\in\tau^r$, and thus by Proposition \ref{prop:homogeneous_equil} any such time is an equilibrium. 
\item If $a>\overline{t}$ then 
\[
\beta(a+n-1)> c(a;a)\ ,
\] 
and so arriving after $a$ is not desirable. This is a sufficient condition for equilibrium as arriving before $a$ is not allowed.
\item If $b<\underline{t}$ then 
\[
-\gamma b> c(b;b)\ ,
\] 
and so arriving a moment before $b$ is not desirable. This is a sufficient condition for equilibrium as arriving after $b$ is not allowed.
\end{enumerate}
\end{proof}

\begin{proof}[Proof of Lemma \ref{lemma:tau_i}]
Recall that when all arrive simultaneously, they are processed in uniform random order and that the cost for customer $i$ is
\[
c_i(t;t)=\sum_{j=0}^{n-1}\frac{1}{n} c_i(t+j)=\frac{1}{n}c_i(t)+\frac{1}{n}\sum_{j=1}^{n-1} c_i(t+j)\ ,
\]
hence,
\[
c_i(t;t)\leq c_i(t)\ \Leftrightarrow \ \frac{1}{n-1}\sum_{j=1}^{n-1} c_i(t+j) \leq c_i(t)\ .
\]
By part (b) of Assumption \ref{assump:ci} $c_i(t)$ is decreasing for all $t<d_i$; therefore by part (c) there exists some $\overline{t}_i$ such that
\begin{equation*}
\frac{1}{n-1}\sum_{j=1}^{n-1} c_i(\overline{t}_i+j) = c_i(\overline{t}_i)\ ,
\end{equation*}
and $\frac{1}{n-1}\sum_{j=1}^{n-1} c_i(t+j) < c_i(t)$ for all $t< \overline{t}_i$, where $\overline{t}_i$ is given by \eqref{eq:ti_up}. 

Similarly,
\[
c_i(t;t)\leq c_i(t+n-1)\ \Leftrightarrow \ \frac{1}{n-1}\sum_{j=0}^{n-2} c_i(t+j) \leq c_i(t+n-1)\ ,
\]
and as $c_i(t)$ is increasing for $t>d_i$ there exists a $\underline{t}_i$ such that
\begin{equation*}
\frac{1}{n-1}\sum_{j=0}^{n-2} c_i(\underline{t}_i+j) = c_i(\underline{t}_i+n-1)\ ,
\end{equation*}
and $\frac{1}{n-1}\sum_{j=0}^{n-2} c_i(t+j) < c_i(t+n-1)$ for all $t>\underline{t}_i$, where $\underline{t}_i$ is given by \eqref{eq:ti_down}.  Furthermore, \eqref{eq:ti_up} implies that
\[
c_i(\overline{t}_i+n-1)>\frac{1}{n-1}\sum_{j=0}^{n-2} c_i(\overline{t}_i+j)\ ,
\]
and we thus conclude that $\underline{t}_i\leq \overline{t}_i$, and therefore $\tau_i=[\underline{t}_i,\overline{t}_i]\neq\emptyset$.
\end{proof}

\begin{proof}[Proof of Proposition \ref{prop:heterogeneous_tau}] $\mbox{}$
\begin{enumerate}[label=(\alph*)] 
\item This is a direct result of Lemma \ref{lemma:tau_i}. Arriving at $t$ is a best response for customer $i$ if all others arrive $t$ for any $t\in\tau_i$, hence if this holds for all $i=1,\ldots,n$ then $t$ is a Nash equilibrium.
\item If $t$ is an equilibrium then
\[
c_i(t)\geq c_i(t;t)=\frac{1}{n}\sum_{j=1}^n c_i(t+j-1),\quad \forall i=1,\ldots,n\ ,
\]
which is equivalent to
\[
\frac{c_i(t)}{\sum_{j=1}^n c_i(t+j-1)}\geq \frac{1}{n},\quad \forall i=1,\ldots,n\ ,
\]
or
\[
\min_{i=1,\ldots,n}\left\lbrace \frac{c_i(t)}{\sum_{j=1}^n c_i(t+j-1)}\right\rbrace \geq \frac{1}{n}\ .
\]
Similarly, the second equilibrium condition,
\[
c_i(t+n-1)\geq c_i(t;t)=\frac{1}{n}\sum_{j=1}^n c_i(t+j-1),\quad \forall i=1,\ldots,n\ ,
\]
yields
\[
\min_{i=1,\ldots,n}\left\lbrace \frac{c_i(t+n-1)}{\sum_{j=1}^n c_i(t+j-1)}\right\rbrace \geq \frac{1}{n}\ .
\]
We conclude that the boundaries of the interval $\tau=[\underline{t},\overline{t}]$ are as stated in the proposition statement.
\end{enumerate}
\end{proof}

\begin{proof}[Proof of Proposition \ref{prop:gamma_i_beta_i}] \mbox{}
For any $-\frac{1}{n-1}<t<0$,
\[
c_i(t+n-1)=\beta_i(t+n-1)\ ,
\]
and
\[
\frac{1}{\beta_i}c_i(t+n-1)=t+n-1=\frac{1}{\beta_k}c_k(t+n-1),\ \forall k\neq i\ .
\]
Condition \eqref{eq:ti_down} is equivalent to 
\[
\frac{1}{\beta_i}c_i(\underline{t}_i+n-1)=\frac{1}{n-1}\left[\sum_{j=i_{\underline{t}_i}+1}^{n-2}(\underline{t}_i+j-1)-\frac{\gamma_i}{\beta_i} \sum_{j=1}^{i_{\underline{t}_i}}(\underline{t}_i+j-1)\right]\ .
\] 
If $\frac{\gamma_k}{\beta_k}\geq\frac{\gamma_i}{\beta_i}$ then
\[
\underline{t}_i+n-1=\frac{1}{\beta_i}c_i(\underline{t}_i+n-1)\geq\left[\sum_{j=i_{\underline{t}_i}+1}^{n-2}(\underline{t}_i+j-1)-\frac{\gamma_k}{\beta_k} \sum_{j=1}^{i_{\underline{t}_i}}(\underline{t}_i+j-1)\right]\ ,
\]
and thus
\[
\underline{t}_k+n-1=\frac{1}{\beta_k}c_i(\underline{t}_k+n-1)=\left[\sum_{j=i_{\underline{t}_k}+1}^{n-2}(\underline{t}_k+j-1)-\frac{\gamma_k}{\beta_k} \sum_{j=1}^{i_{\underline{t}_k}}(\underline{t}_k+j-1)\right]\ ,
\]
implies that $\underline{t}_k\leq \underline{t}_i$. As this argument can be repeated for any pair we have that
\[
\underline{t}_1\leq \underline{t}_2\leq\cdots\leq \underline{t}_n\ .
\]
The same type of argument applied to condition \ref{eq:ti_up} yields
\[
\overline{t}_1\leq \overline{t}_2\leq\cdots\leq \overline{t}_n\ .
\]
We conclude that
\[
\tau^e=\bigcap_{i=1}^n\tau_i=\bigcap_{i=1}^n[\underline{t}_i,\overline{t}_i]=[\underline{t}_n,\overline{t}_1]\ .
\]
\end{proof}

\vspace{-1cm}

\section*{Acknowledgments}
We thank Ran Snitkovsky for his advice and comments. This research was supported by the ISRAEL SCIENCE FOUNDATION (grant No. 355/15).

\bibliographystyle{abbrv}
\small{\bibliography{C:/Users/Liron/Dropbox/University/Research/Full_Bibliography/BigBib}}

\begin{thebibliography}{10}

\bibitem{AS2016d}
E.~Altman and N.~Shimkin.
\newblock The ordered timeline game: Strategic posting times over a temporally
  ordered shared medium.
\newblock {\em Dynamic Games and Applications}, 6(4):429--455, 2016.

\bibitem{APS2014}
A.~Anderson, A.~Park, and L.~Smith.
\newblock Greed, fear, and rushes.
\newblock {\em Available at SSRN 2273777}, 2014.

\bibitem{ASD2014}
R.~Argenziano and P.~Schmidt-Dengler.
\newblock Clustering in n-player preemption games.
\newblock {\em Journal of the European Economic Association}, 12(2):368--396,
  2014.

\bibitem{ADL1993}
R.~Arnott, A.~de~Palma, and R.~Lindsey.
\newblock A structural model of peak-period congestion: A traffic bottleneck
  with elastic demand.
\newblock {\em American Economic Review}, 83(1):161--79, 1993.

\bibitem{AKN2015}
R.~J. Arnott, A.~Kokoza, and M.~Naji.
\newblock A model of rush-hour traffic in an isotropic downtown area.
\newblock {\em CESifo Working Paper Series}, 2015.

\bibitem{AFJMS2015}
Y.~Azar, L.~Fleischer, K.~Jain, V.~Mirrokni, and Z.~Svitkina.
\newblock Optimal coordination mechanisms for unrelated machine scheduling.
\newblock {\em Operations Research}, 63(3):489--500, 2015.

\bibitem{BS1990}
K.~Baker and G.~D. Scudder.
\newblock Sequencing with earliness and tardiness penalties: a review.
\newblock {\em Operations Research}, 38(1):22--36, 1990.

\bibitem{BSO2016}
J.~Breinbjerg, A.~Sebald, and L.~P. {\O}sterdal.
\newblock Strategic behavior and social outcomes in a bottleneck queue:
  experimental evidence.
\newblock {\em Review of Economic Design}, 20(3):207--236, 2016.

\bibitem{BM2010}
M.~K. Brunnermeier and J.~Morgan.
\newblock {Clock games: Theory and experiments}.
\newblock {\em Games and Economic Behavior}, 68(2):532--550, 2010.

\bibitem{BH2007}
Y.~Bukchin and E.~Hanany.
\newblock {Decentralization Cost in Scheduling: A Game-Theoretic Approach}.
\newblock {\em Manufacturing \& Service Operations Management}, 9(3):263--275,
  2007.

\bibitem{CG1994}
C.~Chamley and D.~Gale.
\newblock Information revelation and strategic delay in a model of investment.
\newblock {\em Econometrica}, pages 1065--1085, 1994.

\bibitem{FT2012}
M.~Feldman and T.~Tamir.
\newblock Conflicting congestion effects in resource allocation games.
\newblock {\em Operations Research}, 60(3):529--540, 2012.

\bibitem{FT1985}
D.~Fudenberg and J.~Tirole.
\newblock Preemption and rent equalization in the adoption of new technology.
\newblock {\em The Review of Economic Studies}, 52(3):383--401, 1985.

\bibitem{GTW1988}
M.~R. Garey, R.~E. Tarjan, and G.~T. Wilfong.
\newblock One-processor scheduling with symmetric earliness and tardiness
  penalties.
\newblock {\em Mathematics of Operations Research}, 13(2):330--348, 1988.

\bibitem{G1985}
A.~Glazer.
\newblock The advantages of being first.
\newblock {\em The American Economic Review}, 75(3):473--480, 1985.

\bibitem{GH1983}
A.~Glazer and R.~Hassin.
\newblock {?/M/$\textit{1}$: On the equilibrium distribution of customer
  arrivals}.
\newblock {\em European Journal of Operational Research}, 13(2):146--150, 1983.

\bibitem{GH1987}
A.~Glazer and R.~Hassin.
\newblock Equilibrium arrivals in queues with bulk service at scheduled times.
\newblock {\em Transportation Science}, 21(4):273--278, 1987.

\bibitem{HP1991}
N.~G. Hall and M.~E. Posner.
\newblock {Earliness-tardiness scheduling problems, I: weighted deviation of
  completion times about a common due date}.
\newblock {\em Operations Research}, 39(5):836--846, 1991.

\bibitem{HK2011}
R.~Hassin and Y.~Kleiner.
\newblock Equilibrium and optimal arrival patterns to a server with opening and
  closing times.
\newblock {\em IIE Transactions}, 43(3):164--175, 2011.

\bibitem{HS2005}
R.~Hassin and M.~Shani.
\newblock Machine scheduling with earliness, tardiness and non-execution
  penalties.
\newblock {\em Computers \& Operations Research}, 32(3):683--705, 2005.

\bibitem{H2013}
M.~Haviv.
\newblock When to arrive at a queue with tardiness costs?
\newblock {\em Performance Evaluation}, 70(6):387--399, 2013.

\bibitem{HR2015}
M.~Haviv and L.~Ravner.
\newblock Strategic timing of arrivals to a finite queue multi-server loss
  system.
\newblock {\em Queueing Systems}, 81(1):71--96, 2015.

\bibitem{HJ2015}
H.~Honnappa and R.~Jain.
\newblock Strategic arrivals into queueing networks: the network concert
  queueing game.
\newblock {\em Operations Research}, 63(1):247--259, 2015.

\bibitem{JS2013}
S.~Juneja and N.~Shimkin.
\newblock The concert queueing game: strategic arrivals with waiting and
  tardiness costs.
\newblock {\em Queueing Systems}, 74(4):369--402, 2013.

\bibitem{L1988}
M.~Lambkin.
\newblock Order of entry and performance in new markets.
\newblock {\em Strategic Management Journal}, 9(S1):127--140, 1988.

\bibitem{LvM2004}
M.~A. Lariviere and J.~A. Van~Mieghem.
\newblock {Strategically seeking service: How competition can generate Poisson
  arrivals}.
\newblock {\em Manufacturing \& Service Operations Management}, 6(1):23--40,
  2004.

\bibitem{LW2004}
V.~Lauff and F.~Werner.
\newblock Scheduling with common due date, earliness and tardiness penalties
  for multimachine problems: A survey.
\newblock {\em Mathematical and Computer Modelling}, 40(5):637--655, 2004.

\bibitem{LP2003}
D.~Levin and J.~Peck.
\newblock {To grab for the market or to bide one's time: A dynamic model of
  entry}.
\newblock {\em Rand Journal of Economics}, pages 536--556, 2003.

\bibitem{LM1988}
M.~B. Lieberman and D.~B. Montgomery.
\newblock First-mover advantages.
\newblock {\em Strategic Management Journal}, 9(S1):41--58, 1988.

\bibitem{MH1984}
H.~Mahmassani and R.~Herman.
\newblock Dynamic user equilibrium departure time and route choice on idealized
  traffic arterials.
\newblock {\em Transportation Science}, 18(4):362--384, 1984.

\bibitem{MW2010}
R.~Mason and H.~Weeds.
\newblock Investment, uncertainty and pre-emption.
\newblock {\em International Journal of Industrial Organization},
  28(3):278--287, 2010.

\bibitem{MC2006}
V.~V. Mazalov and J.~V. Chuiko.
\newblock Nash equilibrium in the optimal arrival time problem.
\newblock {\em Computational Technologies}, 11:60--71, 2006.

\bibitem{MM1987}
R.~P. McAfee and J.~McMillan.
\newblock Auctions with a stochastic number of bidders.
\newblock {\em Journal of Economic Theory}, 43(1):1--19, 1987.

\bibitem{PS2008}
A.~Park and L.~Smith.
\newblock Caller number five and related timing games.
\newblock {\em Theoretical Economics}, 3(2):231--256, 2008.

\bibitem{PO2012}
T.~T. Platz and L.~P. {\O}sterdal.
\newblock The curse of the first-in-first-out queue discipline.
\newblock {\em Discussion papers of business and economics, University of
  Southern Denmark}, (10), 2012.

\bibitem{R2014}
L.~Ravner.
\newblock Equilibrium arrival times to a queue with order penalties.
\newblock {\em European Journal of Operational Research}, 239(2):456--468,
  2014.

\bibitem{RHV2016}
L.~Ravner, M.~Haviv, and H.~L. Vu.
\newblock A strategic timing of arrivals to a linear slowdown processor sharing
  system.
\newblock {\em European Journal of Operational Research}, 255(2):496 -- 504,
  2016.

\bibitem{RF1985}
W.~T. Robinson and C.~Fornell.
\newblock Sources of market pioneer advantages in consumer goods industries.
\newblock {\em Journal of Marketing Research}, 22(3):305--317, 1985.

\bibitem{SK2016}
E.~Sherzer and Y.~Kerner.
\newblock When to arrive at a queue with earliness, tardiness and waiting
  costs.
\newblock {\em Submitted}, 2016.

\bibitem{S1974}
J.~M. Smith.
\newblock The theory of games and the evolution of animal conflicts.
\newblock {\em Journal of theoretical biology}, 47(1):209--221, 1974.

\bibitem{V1969}
W.~S. Vickrey.
\newblock Congestion theory and transport investment.
\newblock {\em The American Economic Review}, 59(2):251--260, 1969.

\end{thebibliography}

\end{document}